\documentclass[graybox]{svmult}
\usepackage{amsmath,amssymb,algorithmic, algorithm}

\usepackage{mathptmx}       
\usepackage{helvet}         
\usepackage{courier}        
\usepackage{type1cm}        
\usepackage{makeidx}         
\usepackage{graphicx}        
\usepackage{multicol}        
\usepackage[bottom]{footmisc}
\newcommand{\dv}{{\cal D}}
\newcommand{\jdm}{{\cal J}}
\newtheorem{observation}{Observation}
\begin{document}
\title*{Constructing and Sampling Graphs with a Prescribed Joint Degree Distribution}

\author{Isabelle Stanton \and Ali Pinar }
\institute{Isabelle Stanton \at UC Berkeley, Berkeley, CA, \email{isabelle@eecs.berkeley.edu} 
\and Ali Pinar \at Sandia National Laboratories, Livermore, CA, \email{apinar@sandia.gov}\\ This work was funded by the applied mathematics program at the United States Department of Energy and performed at Sandia National Laboratories, a multiprogram laboratory operated by Sandia Corporation, a wholly owned subsidiary of Lockheed Martin Corporation, for the United States Department of Energy¹s National Nuclear Security Administration under contract DE-AC04-94AL85000.}


\maketitle

\abstract{One of the most influential recent results in network analysis is that many natural networks exhibit a power-law or log-normal degree distribution. This has inspired numerous generative models that match this property. However, more recent work has shown that while these generative models do have the right degree distribution, they are not good models for real life networks due to their differences on other important metrics like conductance. We believe this is, in part, because many of these real-world networks have very different \emph{joint degree distributions}, i.e. the probability that a randomly selected edge will be between nodes of degree $k$ and $l$. Assortativity is a sufficient statistic of the joint degree distribution, and it has been previously noted that social networks tend to be assortative, while biological and technological networks tend to be disassortative.\newline\indent
We suggest understanding the relationship between network structure and the joint degree distribution of graphs is an interesting avenue of further research. An important tool for such studies are algorithms that can generate random instances of graphs with the same joint degree distribution. This is the main topic of this paper and we study the problem from both a theoretical and practical perspective. We provide an algorithm for constructing simple graphs from a given joint degree distribution, and a Monte Carlo Markov Chain method for sampling them. We also show that the state space of simple graphs with a fixed degree distribution is connected via \emph{end point switches}. We empirically evaluate the mixing time of this Markov Chain by using experiments based on the autocorrelation of each edge. These experiments show that our Markov Chain mixes quickly on real graphs, allowing for utilization of our techniques in practice.}


\section{Introduction}
Graphs are widely recognized as the standard modeling language for many complex systems, including physical infrastructure (e.g., Internet, electric power, water, and gas networks), scientific processes (e.g., chemical kinetics, protein interactions, and regulatory networks in biology starting at the gene levels through ecological systems), and relational networks (e.g., citation networks, hyperlinks on the web, and social networks). 
The broader adoption of the graph models over the last decade, along with the growing importance of associated applications, calls for descriptive and generative models for real networks.  What is common among these networks?  How do they differ statistically?  Can we quantify the differences among these networks?  Answering these questions requires understanding the topological properties of these graphs, which  have lead to numerous studies on many  ``real-world'' networks from the Internet to social, biological and technological networks~\cite{fff}. 

Perhaps the most prominent theme in these studies is the skewed degree distribution; real-world graphs have a few vertices with very high degree and many vertices with small degree. There is some dispute as to the exact distribution, some have called it power-law~\cite{bar99,fff}, some log-normal~\cite{amaral,pennock02,mitzsurvey,dgx}, and but all agree that it is `heavy-tailed'~\cite{clauset,sala11}. The ubiquity of this distribution has been a motivator for many different generative models and is often used as a metric for the quality of the model. Models like preferential attachment~\cite{bar99}, the copying model~\cite{KumarRRSTU00}, the Barabasi hierarchical model~\cite{hierbarbasi}, forest-fire model, the Kronecker graph model~\cite{kronecker}, geometric preferential attachment~\cite{FlaxmanFV04} and many more~\cite{conf/kdd/LeskovecKF05,jamming,brst01} study the expected degree distribution and use the results to argue for the strength of their method. Many of these models also match other observed features, such as small diameter or densification~\cite{conf/nips/Kleinberg01}. However, recent studies comparing the generative models with real networks on metrics like conductance~\cite{LeskovecLDM08}, core numbers~\cite{SPK11} and clustering coefficients~\cite{KSP11} show that the models do not match other important features of the networks. 

The degree  distribution alone does not define a graph. McKay's estimate~\cite{mckaysestimate} shows that there may be exponentially many graphs with the same degree distribution.  However,  models based on degree distribution are commonly used to compute statistically significant structures in a graph. For example, the modularity metric for community detection in graphs~\cite{modularity,wmodularity} assumes a null hypothesis for the structure of a graph based on its degree distribution, namely that probability of an edge between vertex $v_i$ and $v_j$ is proportional to $d_id_j$, where $d_i$ and $d_j$ represent  the degrees of vertices $v_i$ and $v_j$.  The modularity of a group of vertices is defined by how much their structure deviates from the null hypothesis, and a higher modularity signifies a better community.  The key point here is that the null hypothesis is solely based on its degree distribution and therefore might be incorrect. Degree distribution based models are also used to predict graph properties~\cite{Mihail02,Aiello00,Chung1,Chung2,Chung3}, benchmark~\cite{Lancichinetti2008}, and analyze the expected run time of algorithms~\cite{CindyJon}.

These studies improve our understanding of the relationship between the degree distribution and the structure of a graph. The shortcomings of these studies give insight into what other features besides the degree distribution would give us a better grasp of a graph's structure. For example, the degree assortativity of a network measure whether nodes attach to other similar or dissimilar vertices. This is not specified by the degree distribution, yet studies have shown that social networks tend to be assortative, while biological and technological networks tend to be dissortative~\cite{mejnassortative1,mejnassortative2}. An example of recent work using assortativity is~\cite{KSP11}. In this study, a high assortativity is assumed for connections that generate high clustering coefficients, and this, in addition to preserving the degree distribution, results in very realistic instances of real-world graphs. Another study that has looked at the joint degree distribution is $dK$-graphs~\cite{dkgraphs}. They propose modeling a graph by looking at the distribution of the structure of all sized $k$ subsets of vertices, where $d=1$ are vertex degrees, $d=2$ are edge degrees (the joint degree distribution), $d=3$ is the degree distribution of triangles and wedges, and so on. It is an interesting idea, as clearly the $nK$ distribution contains all information about the graph, but it is far too detailed as a model. At what $d$ value does the additional information become less useful?

One way to enhance  the results based on degree distribution is to use a more restrictive feature such as the \emph{joint degree distribution}. 
Intuitively, if degree distribution of a graph describes the probability that a vertex selected uniformly at random will be of degree $k$ then its joint degree distribution  describes the probability that a randomly selected \emph{edge} will be between nodes of degree $k$ and $l$. We will use a slightly different concept, the joint degree matrix, where the total number of nodes and edges is specified, and the numbers of edges between each set of degrees is counted. Note that while the joint degree distribution uniquely defines the degree distribution of a graph up to isolated nodes, graphs with the same degree distribution may have very different joint degree distributions. We are not proposing that the joint degree distribution be used as a stand alone descriptive model for generating networks. We believe that understanding the relationship between the joint degree distribution and the network structure is important, and that having the capability to generate random instances of graphs with the same joint degree distribution will help enable this goal. Experiments on real data are valuable, but also drawing conclusions only based on a limited data may be misleading, as the graphs may all be biased the same way.  For a more rigorous study, we need a sampling algorithm that can generate random instances in a reasonable time, which is the motivation of this work.

The primary questions investigated by this paper are: Given a joint degree distribution and an integer $n$, does the joint degree distribution correspond to a real labeled graph? If so, can one construct a graph of size $n$ with that joint degree distribution? Is it possible to construct or generate a \emph{uniformly random} graph with that same joint degree distribution? We address these problems from both a theoretical and from an empirical perspective. In particular, being able to uniformly sample graphs allows one to empirically evaluate which other graph features, like diameter, or eigenvalues, are correlated with the joint degree distribution.

\paragraph{Contributions}
We make several contributions to this problem, both theoretically and experimentally. First, we discuss the necessary and sufficient conditions for a given joint degree vector to be graphical. We prove that these conditions are sufficient by providing a new constructive algorithm. Next, we introduce a new configuration model for the joint degree matrix problem which is a natural extension of the configuration model for the degree sequence problem. Finally, using this configuration model, we develop Markov Chains for sampling both pseudographs and simple graphs with a fixed joint degree matrix. A pseudograph allows multiple edges between two nodes and self-loops. We prove the correctness of both chains and mixing time for the pseudograph chain by using previous work. The mixing time of the simple graph chain is experimentally evaluated using autocorrelation. 

In practice, Monte Carlo Markov Chains are a very popular method for sampling from difficult distributions. However, it is often very difficult to theoretically evaluate the mixing time of the chain, and many practitioners simply stop the chain after 5,000, 10,000 or 20,000 iterations without much justification. Our experimental design with autocorrelation provides a set of statistics that can be used as a justification for choosing a stopping point. Further, we show one way that the autocorrelation technique can be adapted from real-valued samples to combinatorial samples.

\section{Related Work}

The related work can be roughly divided into two categories: constructing and sampling graphs with a fixed degree distribution using sequential importance sampling or Monte Carlo Markov Chain methods, and experimental work on heuristics for generating random graphs with a fixed joint degree distribution.

The methods for constructing graphs with a given degree distribution are primarily either reductions to perfect matchings or sequential sampling methods. There are two popular perfect matching methods. The first is the \emph{configuration model}~\cite{bollobasconfig,conf/stoc/AielloCL00}: $k$ mini-vertices are created for each degree $k$ vertex, and all the mini-vertices are connected. Any perfect matching in the configuration graph corresponds to a graph with the correct degree distribution by merging all of the identified mini-vertices. This allows multiple edges and self-loops, which are often undesirable. See Figure~\ref{fig:configmodel}. The second approach, the \emph{gadget configuration model}, prevents multi-edges and self-loops by creating a gadget for each vertex. If $v_i$ has degree $d_i$, then it is replaced with a complete bipartite graph $(U_i,V_i)$ with $|U_i|=n-1-d_i$ and $|V_i|=n-1$. Exactly one node in each $V_i$ is connected to each other $V_j$, representing edge $(i,j)$~\cite{ktv}. Any perfect matching in this model corresponds exactly to a simple graph by using the edges in the matching that correspond with edges connecting any $V_i$ to any $V_j$. We use a natural extension of the first configuration model to the joint degree distribution problem.
\begin{figure}[ht]
\centering
\includegraphics[width=50mm]{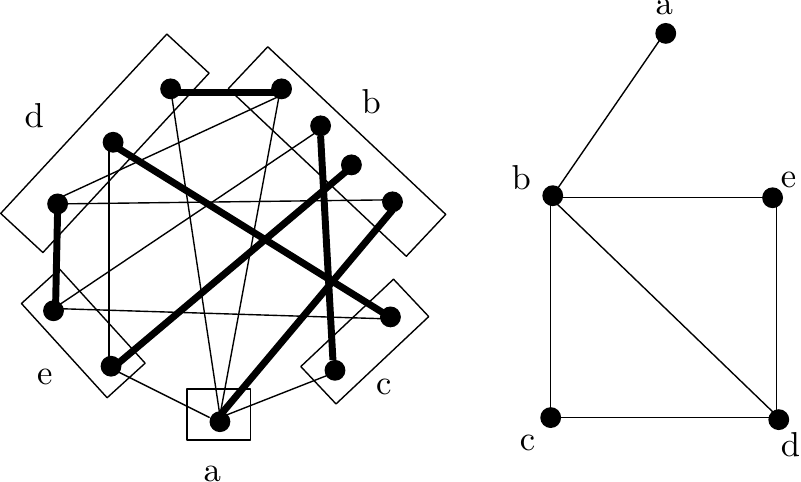}
\caption{On the left, we see an example of the configuration model of the degree distribution of the graph on the right. The edges corresponding to that graph are bold. Each vertex is split into a number of mini-vertices equal to its degree, and then all mini-vertices are connected. Not all edges are shown for clarity.}
\label{fig:configmodel}
\end{figure}

There are also sequential sampling methods that will construct a graph with a given degree distribution. Some of these are based on the necessary and sufficient Erd\H{o}s-Gallai conditions for a degree sequence to be graphical~\cite{blitzsteindiaconis}, while others follow the method of Steger and Wormald~\cite{bayatikimsaberi,stegerwormald,SinclairJ89,JerrumS90,KimV06}. These combine the construction and sampling parts of the problem and can be quite fast. The current best work can sample graphs where $d_{max}=O(m^{1/4-\tau})$ in $O(md_{max})$ time~\cite{bayatikimsaberi}.

Another approach for sampling graphs with a given degree distribution is to use a Monte Carlo Markov Chain method. There is significant work on sampling perfect matchings~\cite{JerrumSV04,Broder86}. There has also been work specifically targeted at the degree distribution problem. Kannan, Tetali and Vempala~\cite{ktv} analyze the mixing time of a Markov Chain that mixes on the configuration model, and another for the gadget configuration model. Gkantsidis, Mihail and Zegura~\cite{GkantsidisMMZ03} use a Markov Chain on the configuration model, but reject any transition that creates a self-loop, multiple edge or disconnects the graph. Both of these chains use the work of Taylor~\cite{taylor} to argue that the state space is connected. 


Amanatidis, Green and Mihail study the problems of when a given joint degree matrix has graphical representation and, further, when it has connected graphical representation~\cite{amanatidis}. They give necessary and sufficient conditions for both of these problems, and constructive algorithms. In Section 2, we give a simpler constructive algorithm for creating a graphical representation that is based on solving the degree sequence problem instead of alternating structures.

Another vein of related work is that of Mahadevan et al. who introduce the concept of $dK$-series~\cite{dkgraphs,orbis}. In this model, $d$ refers to the dimension of the distribution and $2K$ is the joint degree distribution. They propose a heuristic for generating random $2K$-graphs for a fixed $2K$ distribution via edge rewirings. However, their method can get stuck if there exists a degree in the graph for which there is only 1 node with that degree. This is because the state space is not connected. We provide a theoretically sound method of doing this.
 
Finally, Newman also studies the problem of fixing an assortativity value, finding a \emph{joint remaining degree distribution} with that value, and then sampling a random graph with that distribution using Markov Chains~\cite{mejnassortative1,mejnassortative2}. His Markov Chain starts at any graph with the correct degree distribution and converges to a pseudograph with the correct joint remaining degree distribution. By contrast, our work provides a theoretically sound way of constructing a simple graph with a given joint degree distribution first, and our Markov Chain only has simple graphs with the same joint degree distribution as its state space.
\section{Notation and Definitions}

 Formally, a degree distribution of a graph is the probability that a node chosen at random will be of degree $k$. Similarly, the joint degree distribution is the probability that a randomly selected \emph{edge} will have end points of degree $k$ and $l$. In this paper, we are concerned with constructing graphs that exactly match these distributions, so rather than probabilities, we will use a counting definition below and call it the \emph{joint degree matrix}. In particular, we will be concerned with generating \emph{simple} graphs that do not contain multiple edges or self-loops. Any graph that may have multiple edges or self loops will be referred to as a pseudograph.

\begin{definition} The degree vector (DV) $d(G)$ of a graph $G$ is a vector where $d(G)_k$ is the number of nodes of degree $k$ in $G$.
\end{definition}

A generic degree vector will be denoted by $\dv$.

\begin{definition} The joint degree matrix (JDM) $\jdm(G)$ of a graph $G$ is a matrix where $\jdm(G)_{k,l}$ is exactly the number of edges between nodes of degree $k$ and degree $l$ in $G$.
\end{definition}


A generic joint degree matrix will be denoted by $\jdm$. Given a joint degree matrix, $\jdm$, we can recover the number of edges in the graph as $m = \sum_{k=1}^{\infty}\sum_{l=k}^{\infty} \jdm_{k,l}$. We can also recover the degree vector as $\dv_k = \frac 1 k (\jdm_{k,k} + \sum_{l=1}^{\infty} \jdm_{k,l})$. The term $\jdm_{k,k}$ is added twice because $k\dv_k$ is the number of end points of degree $k$ and the edges in $\jdm_{k,k}$ contribute two end points.

The number of nodes, $n$ is then $\sum_{k=1}^{\infty} \dv_k$. This count does not include any degree 0 vertices, as these have no edges in the joint degree matrix. Given $n$ and $m$, we can easily get the degree distribution and joint degree distribution. They are $P(k)=\frac 1 n \dv_k$ while $P(k,l) = \frac 1 m \jdm_{k,l}$. Note that $P(k)$ is not quite the marginal of $P(k,l)$ although it is closely related. 

\paragraph{The Joint Degree Matrix Configuration Model} We propose a new configuration model for the joint degree distribution problem. Given $\jdm$ and its corresponding $\dv$ we create $k$ labeled mini-vertices for every vertex of degree $k$. In addition, for every edge with end points of degree $k$ and $l$ we create two labeled mini-end points, one of class $k$ and one of class $l$. We connect all degree $k$ mini-vertices to the class $k$ mini-end points. This forms a complete bipartite graph for each degree, and each of these forms a connected component that is disconnected from all other components. We will call each of these components the ``$k$-neighborhood''.  Notice that there are $k\dv_k$ mini-vertices of degree $k$, and $k\dv_k=\jdm_{k,k}+\sum_l \jdm_{k,l}$ corresponding mini-end points in each $k$-neighborhood. This is pictured in Figure~\ref{fig:jdmconfig}. Take any perfect matching in this graph. If we merge each pair of mini-end points that correspond to the same edge, we will have some pseudograph that has exactly the desired joint degree matrix. This observation forms the basis of our sampling method.

\begin{figure}[ht]
\centering
\includegraphics[width=80mm]{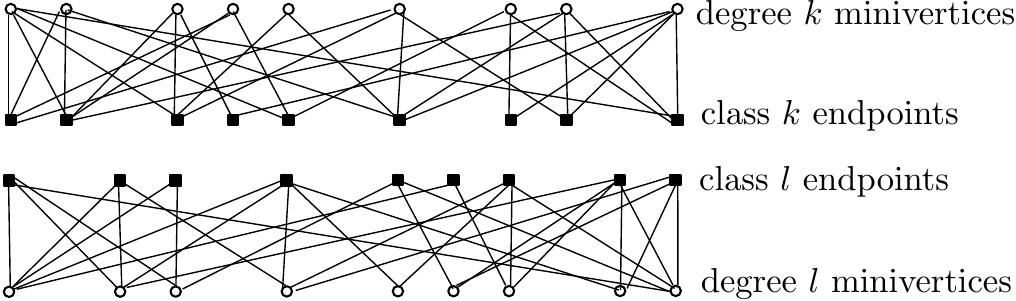}
\caption{The joint degree matrix configuration model. This shows just two degree neighborhoods of the joint degree matrix configuration model. Each vertex of degree $k$ is split into $k$ mini-vertices which are represented by the circles. These then form a complete bipartite component when they are connected with the class $k$ end points, the squares. Each degree neighborhood is completely disconnected from all others. Not all edges are included for clarity.}
\label{fig:jdmconfig}
\end{figure}

\section{Constructing Graphs with a Given Joint Degree Matrix}\label{sec:construct}

The Erd\H{o}s-Gallai condition is a necessary and sufficient condition for a degree sequence to be realizable as a simple graph. 
\begin{theorem} {\bf Erd\H{o}s-Gallai} A degree sequence $\overline{d}=\{d_1,d_2,\cdots d_n\}$ sorted in non-increasing order is graphical if and only if for every $k\leq n$, $\sum_{i=1}^k d_i \leq k(k-1) + \sum_{i=k+1}^n \min(d_i,k)$.
\end{theorem}

The necessity of this condition comes from noting that in a set of vertices of size $k$, there can be at most $k \choose 2$ internal edges, and for each vertex $v$ not in the subset, there can be at most $\min\{d(v),k\}$ edges entering. The condition considers each subset of decreasing degree vertices and looks at the degree requirements of those nodes. If the requirement is more than the available edges, the sequence cannot be graphical. The sufficiency is shown via the constructive Havel-Hakimi algorithm~\cite{havel,hakimi}.

The existence of the Erd\H{o}s-Gallai condition inspires us to ask whether similar necessary and sufficient conditions exist for a joint degree matrix to be graphical. The following necessary and sufficient conditions were independently studied by Amanatidis et al.~\cite{amanatidis}. 
\begin{theorem}
Let $\jdm$ be given and $\dv$ be the associated degree distribution.
$\jdm$ can be realized as a simple graph if and only if (1) $\dv_k$ is integer-valued for all $k$ and (2) $\forall k,l,$ if $k\neq l$ then $\jdm_{k,l}\leq \dv_k\dv_l$. For each $k$, $\jdm_{k,k}\leq {\dv_k\choose 2}$.
\end{theorem}

The necessity of these conditions is clear. The first condition requires that there are an integer number of nodes of each degree value. The next two are that the number of edges between nodes of degree $k$ and $l$ (or $k$ and $k$) are not more than the total possible number of $k$ to $l$ edges in a simple graph defined by the marginal degree sequences. Amanatidis et al. show the sufficiency through a constructive algorithm. We will now introduce a new algorithm that runs in $O(m)$ time.

The algorithm proceeds by building a nearly regular graph for each class of edges, $\jdm_{k,l}$. Assume that $k\neq l$ for simplicity. Each of the $\dv_k$ nodes of degree $k$ receives $\lfloor \jdm_{k,l}/ \dv_k\rfloor$ edges, while $\jdm_{k,l}\mod \dv_k$ each have an extra edge. Similarly, the $l$ degree nodes have $\lfloor \jdm_{k,l}/\dv_l\rfloor$ edges, with $\jdm_{k,l}\mod \dv_l$ having 1 extra. We can then construct a simple bipartite graph with this degree sequence. This can be done in linear time in the number of edges using queues as is discussed after Lemma~\ref{obs:nearlyequal}. If $k = l$, the only differences are that the graph is no longer bipartite and there are $2\jdm_{k,k}$ end points to be distributed among $\dv_k$ nodes. To find a simple nearly regular graph, one can use the Havel-Hakimi~\cite{hakimi,havel} algorithm in $O(\jdm_{k,k})$ time by using the degree sequence of the graph as input to the algorithm.

We must show that there is a way to combine all of these nearly-regular graphs together without violating any degree constraints. Let $d=\langle d_1,d_2,\cdots d_n\rangle$ be the sorted non-increasing order degree sequence from $\dv$. Let $\hat{d}_v$ denote the residual degree sequence where the residual degree of a vertex $v$ is $d_v$ minus the number of edges that currently neighbor $v$. Also, let $\hat{\dv}_k$ denote the number of nodes of degree $k$ that have non-zero residual degree, i.e. $\hat{\dv}_k=\sum_{d_j=k}{\mathbf 1}(\hat{d_j}\neq 0)$.

\begin{algorithm}

\begin{algorithmic}[1]
\caption{Greedy Graph Construction with a Fixed JDM, Input: $\jdm$, $n$, $m$, $\dv$}
\FOR{$k=n\cdots 1$ {\bf and}  $l=k\cdots 1$}
			\IF {$k\neq l$}
				\STATE Let $a = \jdm_{k,l}\mod \dv_k$ and $b=\jdm_{k,l}\mod \dv_l$
				\STATE Let $x_1\cdots x_{a} = \lfloor \frac{\jdm_{k,l}}{\dv_k}\rfloor+1$, $x_{a+1}\cdots x_{\dv_k} 	=\lfloor \frac{\jdm_{k,l}}{\dv_k}\rfloor$ and $y_1\cdots y_{b} = \lfloor \frac{\jdm_{k,l}}{\dv_l}\rfloor+1$, $y_{b+1}\cdots y_{\dv_l} =\lfloor \frac{\jdm_{k,l}}{\dv_l}\rfloor$
				\STATE Construct a simple bipartite graph $B$ with degree sequence $x_1\cdots x_{\dv_k},y_1\cdots y_{\dv_l}$
			\ELSE
				\STATE Let $c = 2\jdm_{k,k}\mod \dv_k$
				\STATE Let $x_1\cdots x_{c} = \lfloor \frac{2\jdm_{k,k}}{\dv_k}\rfloor+1$ and $x_{c+1}\cdots x_{\dv_k} = \lfloor \frac{2\jdm_{k,k}}{\dv_k}\rfloor$
				\STATE Construct a simple graph $B$ with the degree sequence $x_1\cdots x_{\dv_k}$
			\ENDIF
			\STATE Place $B$ into $G$ by matching the nodes of degree $k$ with higher residual degree with $x_1\cdots x_{a}$ and those of degree $l$ with higher residual degree with $y_1\cdots y_{b}$. The other vertices in $B$ can be matched in any way with those in $G$ of degree $k$ and $l$ 
			\STATE Update the residual degrees of each $k$ and $l$ degree node.
\ENDFOR
\end{algorithmic}\label{alg:greedy}
\end{algorithm}

To combine the nearly uniform subgraphs, we start with the largest degree nodes, and the corresponding largest degree classes. It is not necessary to start with the largest, but it simplifies the proof. First, we note that after every iteration, the joint degree sequence is still feasible if $\forall k,l, k\neq l$ $\hat{\jdm}_{k,l}\leq \hat{\dv}_k\hat{\dv}_l$ and $\forall k$ $\hat{\jdm}_{k,k}\leq {\hat{\dv}_k \choose 2}$.

We will prove that Algorithm~\ref{alg:greedy} can always satisfy the feasibility conditions. First, we note a fact.

\begin{observation} For all $k$, $\sum_{l} \hat{\jdm}_{k,l} +\hat{\jdm}_{k,k} = \sum_{d_j=k}\hat{d_j}$ \label{sums}\end{observation}

This follows directly from the fact that the left hand side is summing over all of the $k$ end points needed by $\hat{\jdm}$ while the right hand side is summing up the available residual end points from the degree distribution. Next, we note that if all residual degrees for degree $k$ nodes are either 0 or 1, then:

\begin{observation} If, for all $j$ such that $d_j=k$, $\hat{d_j}=0$ or 1 then\\ $\sum_{d_j=k}\hat{d_j} = \sum_{d_j=k}{\mathbf 1}(\hat{d_j} \neq 0)=\hat{\dv}_k$.\label{indicator}\end{observation}


\begin{lemma} After every iteration, for every pair of vertices $u,v$ of any degree $k$,\\ $|\hat{d}_u-\hat{d}_v|\leq 1$.\label{obs:nearlyequal}
\end{lemma}

Amanatidis et al. refer to Lemma~\ref{obs:nearlyequal} as the \emph{balanced degree invariant}. This is most easily proven by considering the vertices of degree $k$ as a queue. If there are $x$ edges to be assigned, we can consider the process of deciding how many edges to assign each vertex as being one of popping vertices from the top of the queue and reinserting them at the end $x$ times. Each vertex is assigned edges equal to the number of times it was popped. The next time we assign edges with end points of degree $k$, we start with the queue at the same position as where we ended previously. It is clear that no vertex can be popped twice without all other vertices being popped at least once.

\begin{lemma} 
The above algorithm can always greedily produce a graph that satisfies $\jdm$, provided $\jdm$ satisfies the initial necessary conditions.\label{lem:algworks}
\end{lemma}

\begin{proof} 
There is one key observation about this algorithm - it maximizes $\hat{\dv}_k\hat{\dv}_l$ by ensuring that the residual degrees of any two vertices of the same degree never differ by more than 1. By maximizing the number of available vertices, we can not get stuck adding a self-loop or multiple edge. From this, we gather that if, for some degree $k$, there exists a vertex $j$ such that $\hat{d}_j = 0$, then for all vertices of degree $k$, their residuals must be either 0 or 1. This means that $\sum_{d_j=k}\hat{d_j} = \hat{\dv}_k\geq \hat{\jdm}_{k,l}$ for every other $l$ from Observation~\ref{indicator}.

From the initial conditions, we have that for every $k,l$ $\jdm_{k,l}\leq \dv_k\dv_l$. $\dv_k=\hat{\dv}_k$ provided that all degree $k$ vertices have non-zero residuals. Otherwise, for any unprocessed pair, $\jdm_{k,l}\leq \min\{\hat{\dv}_k,\hat{\dv}_l\}\leq \hat{\dv}_k\hat{\dv}_l$. For the $k,k$ case, it is clear that $\jdm_{k,k}\leq \hat{\dv}_k\leq {\hat{\dv}_k\choose 2}$. Therefore, the residual joint degree matrix and degree sequence will always be feasible, and the algorithm can always continue.\qed \end{proof}

A natural question is that since the joint degree distribution contains all of the information in the degree distribution, do the joint degree distribution necessary conditions easily imply the Erd\H{o}s-Gallai condition? This can easily be shown to be true.

\begin{corollary} The necessary conditions for a joint degree matrix to be graphical imply that the associated degree vector satisfies the Erd\H{o}s-Gallai condition.\label{thm:egcond}
\end{corollary}

\section{Uniformly Sampling Graphs with Monte Carlo Markov Chain (MCMC) Methods}\label{sec:mcmctheory}

We now turn our attention to uniformly sampling graphs with a given graphical joint degree matrix using MCMC methods. We return to the joint degree matrix configuration model. We can obtain a starting configuration for any graphical joint degree matrix by using Algorithm 1. This configuration consists of one complete bipartite component for each degree with a perfect matching selected. The transitions we use select any end point uniformly at random, then select any other end point in its degree neighborhood and swap the two edges that these neighbor. In Figure~\ref{fig:jdmconfig}, this is equivalent to selecting one of the square endpoints uniformly at random and then selecting another uniformly at random from the same connected component and then swapping the edges. A more complex version of this chain checks that this swap does not create a multiple edge or self-loop. Formally, the transition function is a randomized algorithm given by Algorithm~\ref{alg:transition}.
\begin{algorithm}
\caption{Markov Chain Transition Function, Input: a configuration $C$}
\begin{algorithmic}[1]
\STATE With probability $0.5$, stay at configuration $C$. Else:
\STATE Select any endpoint $e_1$ uniformly at random. It neighbors a vertex $v_1$ in configuration $C$
\STATE Select any $e_2$ u.a.r from $e_1$'s degree neighborhood. It neighbors $v_2$
\STATE (Optional: If the graph obtained from the configuration with edges $E\cup\{(e_1,v_2),(e_2,v_1)\}\setminus \{(e_1,v_1),(e_2,v_2)\}$ contains a multi-edge or self-loop, reject)
\STATE $E \leftarrow E\cup\{(e_1,v_2),(e_2,v_1)\}\setminus \{(e_1,v_1),(e_2,v_2)\}$
\end{algorithmic}\label{alg:transition}
\end{algorithm}

There are two chains described by Algorithm~\ref{alg:transition}. The first, $\cal A$ doesn't have step (4) and its state space is all pseudographs with the desired joint degree matrix. The second, $\cal B$ includes step (4) and only transitions to and from simple graphs with the correct joint degree matrix.

We remind the reader of the standard result that any irreducible, aperiodic Markov Chain with symmetric transitions converges to the uniform distribution over its state space. Both $\cal A$ and $\cal B$ are aperiodic, due to the self-loop to each state. From the description of the transition function, we can see that $\cal A$ is symmetric. This is less clear for the transition function of $\cal B$. Is it possible for two connected configurations to have a different number of feasible transitions in a given degree neighborhood? We show that it is not the case in the following lemma.

\begin{lemma} The transition function of $\cal B$ is symmetric.\label{lem:symmetric}
\end{lemma}

\begin{proof}
Let $C_1$ and $C_2$ be two neighboring configurations in $\cal B$. This means that they differ by exactly 4 edges in exactly 1 degree neighborhood. Let this degree be $k$ and let these edges be $e_1v_1$ and $e_2v_2$ in $C_1$ whereas they are $e_1v_2$ and $e_2v_1$ in $C_2$. We want to show that $C_1$ and $C_2$ have exactly the same number of feasible $k$-degree swaps.

Without loss of generality, let $e_x,e_y$ be a swap that is prevented by $e_1$ in $C_1$ but allowed in $C_2$. This must mean that $e_x$ neighbors $v_1$ and $e_y$ neighbors some $v_y\neq v_1,v_2$. Notice that the swap $e_1e_x$ is currently feasible. However, in $C_2$, it is now infeasible to swap $e_1,e_x$, even though $e_x$ and $e_y$ are now possible.

If we consider the other cases, like $e_x,e_y$ is prevented by both $e_1$ and $e_2$, then after swapping $e_1$ and $e_2$, $e_x,e_y$ is still infeasible. If swapping $e_1$ and $e_2$ makes something feasible in $C_1$ infeasible in $C_2$, then we can use the above argument in reverse. This means that the number of feasible swaps in a $k$-neighborhood is invariant under $k$-degree swaps.\qed
\end{proof}

The remaining important question is the connectivity of the state space over these chains. It is simple to show that the state space of $\cal A$ is connected. We note that it is a standard result that all perfect matchings in a complete bipartite graph are connected via edge swaps~\cite{taylor}. Moreover, the space of pseudographs can be seen exactly as the set of all perfect matchings over the disconnected complete bipartite degree neighborhoods in the joint degree matrix configuration model. The connectivity result is much less obvious for $\cal B$. We adapt a result of Taylor~\cite{taylor} that all graphs with a given degree sequence are connected via edge swaps in order to prove this. The proof is inductive and follows the structure of Taylor's proof.

\begin{theorem} Given two simple graphs, $G_1$ and $G_2$ of the same size with the same joint degree matrix, there exists a series of endpoint rewirings to transform $G_1$ into $G_2$ (and vice versa) where every intermediate graph is also simple.\label{thm:swaps}
\end{theorem}

\begin{proof}
This proof will proceed by induction on the number of nodes in the graph. The base case is when there are 3 nodes. There are 3 realizable JDMs. Each is uniquely realizable, so there are no switchings available.

\begin{figure}[ht]
\centering\includegraphics[width=5cm]{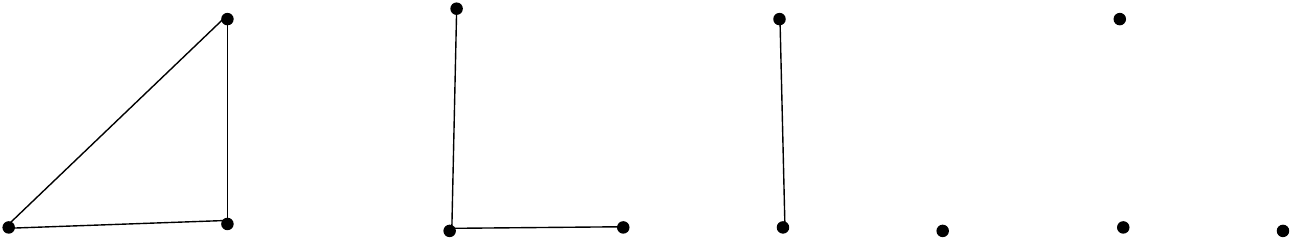}
\caption{The three potential joint degree distributions when $n=3$.}
\label{fig:basecase}
\end{figure}

Assume that this is true for $n=k$. Let $G_1$ and $G_2$ have $k+1$ vertices. Label the nodes of $G_1$ and $G_2$ $v_1\cdots v_{k+1}$ such that $deg(v_1)\geq deg(v_2)\geq \cdots \geq deg(v_{k+1})$. Our goal will be to show that both graphs can be transformed in $G_1'$ and $G_2'$ respectively such that $v_1$ neighbors the same nodes in each graph, and the transitions are all through simple graphs. Now we can remove $v_1$ to create $G_1''$ and $G_2''$, each with $n-1$ nodes and identical JDMs. By the inductive hypothesis, these can be transformed into one other and the result follows.

\begin{figure}[H]
\begin{minipage}[b]{0.45\linewidth}\centering
\includegraphics[width=0.5\columnwidth]{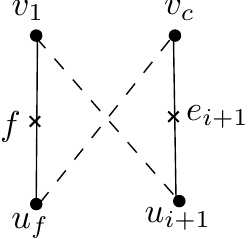}
\caption{The dotted edges represent the troublesome edges that we may need to swap out before we can swap $v_1$ and $v_c$. }
\label{fig:case2}
\end{minipage}\hfill
\begin{minipage}[b]{0.48\linewidth}\centering
\includegraphics[width=.9\columnwidth]{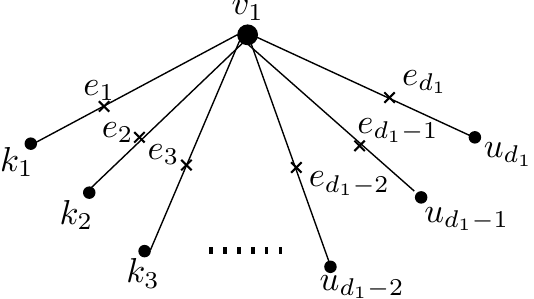}
\caption{The disk is $v_1$. The crosses are $e_1\cdots e_{d_1}$.}\label{fig:lnodes}
\end{minipage}
\end{figure}

We will break the analysis into two cases. For both cases, we will have a set of target edges, $e_1,e_2\cdots e_{d_1}$ that we want $v_1$ to be connected to. Without loss of generality, we let this set be the edges that $v_1$ currently neighbors in $G_2$. We assume that the edges are ordered in reverse lexicographic order by the degrees of their endpoints. This will guarantee that the resulting construction for $v_1$ is graphical and that we have a non-increasing ordering on the requisite endpoints. Now, let $k_i$ denote the endpoint in $G_2$ for edge $e_i$ that isn't $v_1$. 

\noindent {\bf Case 1) } For the first case, we will assume that $v_1$ is already the endpoint of all edges $e_1,e_2\cdots e_{d_1}$ but that all of the $k_i$ may not be assigned correctly as in Figure~\ref{fig:lnodes}. Assume that $e_1,e_2\cdots e_{i-1}$ are all edges $(v_1,k_1)\cdots (v_1,k_{i-1})$ and that $e_i$ is the first that isn't matched to its appropriate $k_i$. 

Call the current endpoint of the other endpoint of $e_i$ $u_i$. We know that $deg(k_i)=deg(u_i)$ and that $k_i$ currently neighbors $deg(k_i)$ other nodes, $\Gamma(k_i)$. We have two cases here. One is that $v_1\in \Gamma(k_i)$ but via edge $f$ instead of $e_i$. Here, we can swap $v_1$ on the endpoints of $f$ and $e_i$ so that the edge $v_1-e_i-k_i$ is in the graph. $f$ can not be an $e_j$ where $j<i$ because those edges have their correct endpoints, $k_j$ assigned. This is demonstrated in Figure~\ref{fig:simplecase1}.

The other case is that $v_1\not \in \Gamma(k_i)$. If this is the case, then there must exist some $x\in \Gamma(k_i)\setminus \Gamma(u_i)$ because $d(u_i)=d(k_i)$ and $u_i$ neighbors $v_1$ while $k_i$ doesn't.  Therefore, we can swap the edges $v_1-e_i-u_i$ and $x-f-k_i$ to $v_1-e_i-k_i$ and $x-f-u_i$ without creating any self-loops or multiple edges. This is demonstrated in Figure~\ref{fig:simplecase1}.

\begin{figure}[H]
\begin{minipage}[b]{0.48\linewidth}\centering
\centering\includegraphics[width=\columnwidth]{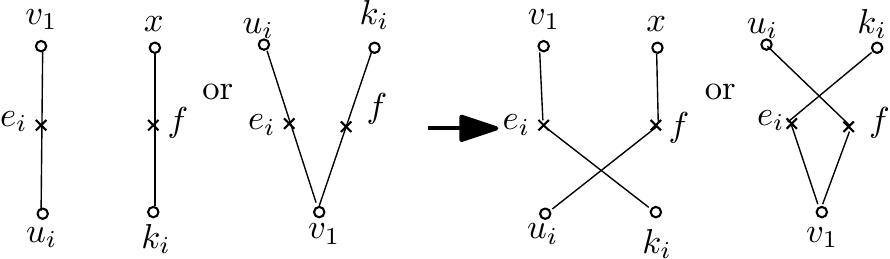}
\caption{The two parts of Case (1).}
\label{fig:simplecase1}
\end{minipage}\hfill
\begin{minipage}[b]{0.48\linewidth}\centering
\centering\includegraphics[width=\columnwidth]{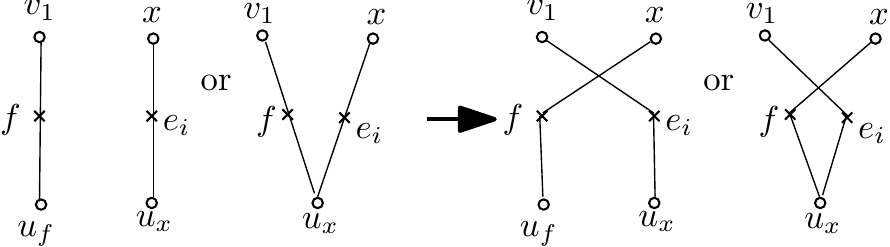}
\caption{The two parts of Case (2)}
\label{fig:simplecase2}
\end{minipage}
\end{figure}

Therefore, we can swap all of the correct endpoints onto the correct edges.

{\bf Case 2) } For the second case, we assume that the edges $e_1,\cdots e_{d_1}$ are distributed over $l$ nodes of degree $d_1$. We want to show that we can move all of the edges $e_1\cdots e_{d_1}$ so that $v_1$ is an endpoint. If this is achievable, we have exactly Case 1.

Let $e_1,\cdots e_{i-1}$ be currently matched to $v_i$ and let $e_i$ be matched to some $x$ such that $deg(x)=d_1$. Let $f$ be an edge currently matched to $v_1$ that is not part of $e_1\cdots e_{d_1}$ and let its other endpoint be $u_f$. Let the other end point of $e_i$ be $u_{x}$ as in Figure~\ref{fig:simplecase2}.

We now have several initial cases that are all easy to handle. First,
if $v,x,u_x,u_f$ are all distinct and $(v,u_x)$ and $(x,u_f)$ are not edges then we can easily swap $v$ and $x$ such that the edges go from $v-f-u_f$ and $x-e_i-u_x$ to $v-e_i-u_x$ and $x-f-u_f$.
Next, if $u_f=u_x$ then we can simply swap $v_1$ onto $e_i$ and $x$ onto $f$ and, again, $v_1$ will neighbor $e_i$. This will not create any self-loops or multiple edges because the graph itself will be isomorphic. This situations are both shown in Figure~\ref{fig:simplecase2}.

The next case is that $x=u_f$. If we try to swap $v_1$ onto $e_i$ then we create a self-loop from $x$ to $x$ via $f$. Instead, we note that since the JDM is graphical, there must exist a third vertex $y$ of the same degree as $v_1$ and $x$ that does not neighbor $x$. Now, $y$ neighbors an edge $g$, and we can swap $x-f$ and $y-g$ to $x-g$ and $y-f$. The edges are $v_1-f-y$ and $x-e_i-u_i$ and $e_i$ can be swapped onto $v_1$ without conflict.

The cases left to analyze are those where the nodes are all distinct and $(v_1,u_x)$ or $(x,u_f)$ are edges in the graph. We will analyze these separately.

{\bf Case 2a) } If $(v_1,u_x)$ is an edge in the graph, then it must be so through some edge named $g$. Note that this means we have $v_1-g-u_x$ and $x-e_i-u_x$. We can swap this to $v_1-e_i-u_x$ and $x-g-u_x$ and have an isomorphic graph provided that $g$ is not some $e_j$ where $j<i$. This is the top case in Figure~\ref{fig:simplecase2a}.

If $g$ is some $e_j$ then it must be that $u_x=k_j$. This is distinct from $k_i$. $deg(k_j)=deg(k_i)$ so there must exist some edge $h$ that $k_i$ neighbors with its other endpoint being $y$. There are again three cases, when $y\neq x,v_1$ $y=x$ and when $y=v_1$. These are the bottom three rows illustrated in Figure~\ref{fig:simplecase2a}. The first is the simplest. Here, we can assume that $k_j$ does not neighbor $y$ (because it neighbors $v_1$ and $x$ that $k_i$ does not) so we can swap $k_j$ onto $h$ and $k_i$ onto $e_1$. This has removed the offending edge, and we can now swap $v_1$ onto $e_1$ and $x$ onto $f$.

When $y=x$, we first swap $k_i$ onto $e_j$ and $k_j$ onto $h$. Next, we swap $v$ onto $e_i$ and $x$ onto $f$ as they no longer share an offending edge.

Finally, when $y=v_1$, we use a sequence of three swaps. The first is $k_i$ onto $e_j$ and $k_j$ onto $h$. The next is $v_1$ onto $e_1$ and $x$ onto $h$. Finally, we swap $k_j$ back onto $e_j$ and $k_i$ onto $e_i$.

\begin{figure}
\centering\includegraphics[width=0.7\columnwidth]{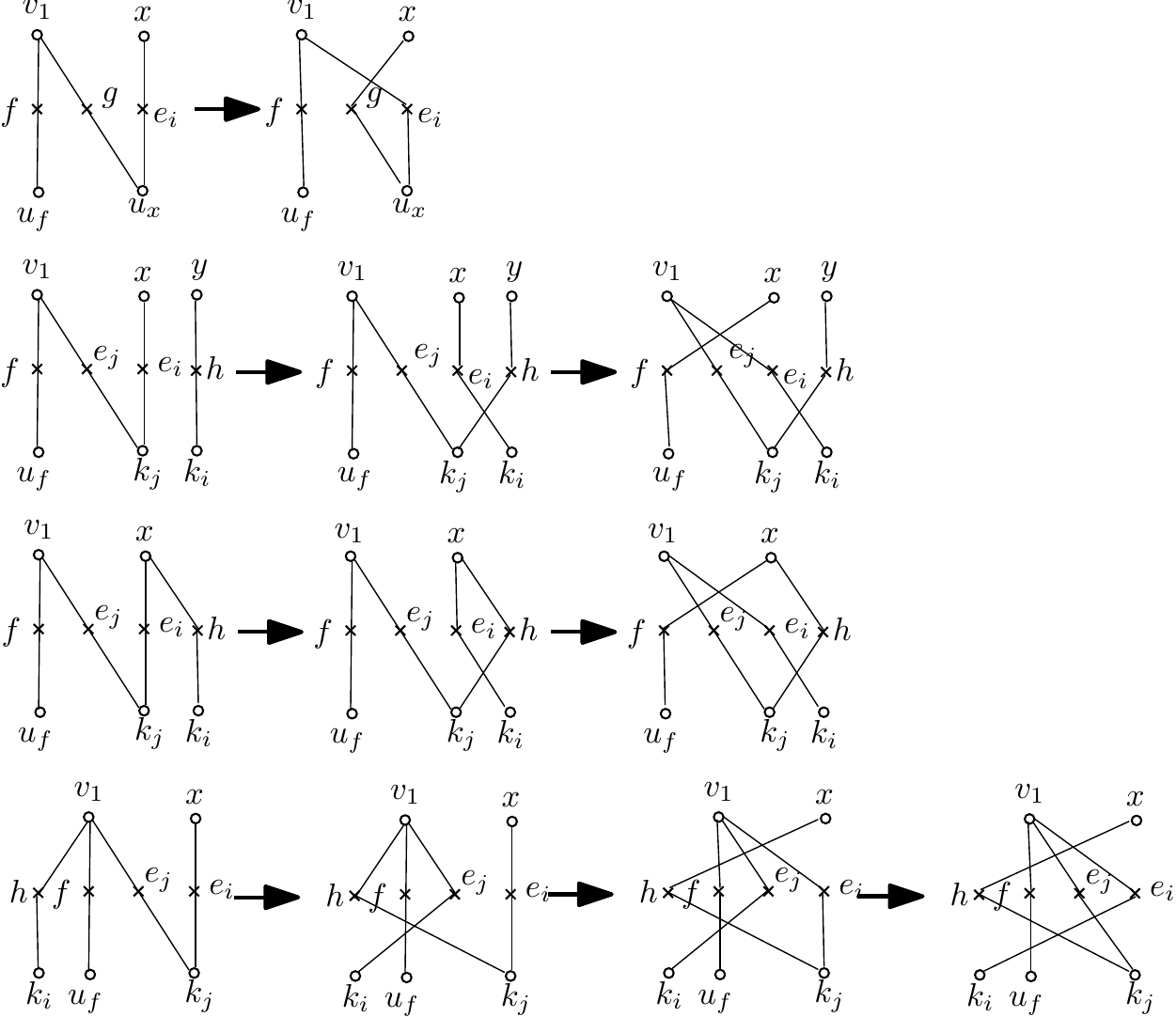}
\caption{A graphical representation of the situations discussed in Case (2a).}
\label{fig:simplecase2a}
\end{figure}

{\bf Case 2b) } If $(x,u_f)$ is an edge in the graph, then it must be through some edge $g$ such that $x-g-u_f$ and $x-e_i-u_x$. Without loss of generality, assume that $f$ is the only edge neighboring $v_1$ that isn't an $e_j$. Since $f$ doesn't neighbor $v_1$ in $G_2$, there must either exist a $w$ with $deg(w)=deg(u_f)$ or $v_s$ with $deg(v_s)=d(v_1)$. This relies critically upon the fact that $f$ and $g$ are the same class edge. If there is a $w$, then it doesn't neighbor $v_1$ (or we can apply the above argument to find a $w'$) and it must have some neighbor $y\in \Gamma(w)\setminus \Gamma(u)$ through edge $h$. Therefore, we can swap $u_f$ onto $h$ and $w$ onto $f$. This removes the offending edge, and we can now swap $v_1$ onto $e_i$ and $x$ onto $f$.

If $v_s$ exists instead, then by the same argument, there exists some edge $h$ with endpoint $u_s$ such that $v_s\notin \Gamma(u_f)$ and $u_s\notin \Gamma(x)$. Therefore, we can swap $v_s-h$ and $x-g$ to $v_s-g$ and $x-h$. This again removes the troublesome edge and allows us to swap $v_1$ onto $e_i$.

\begin{figure}
\centering\includegraphics[width=0.6\columnwidth]{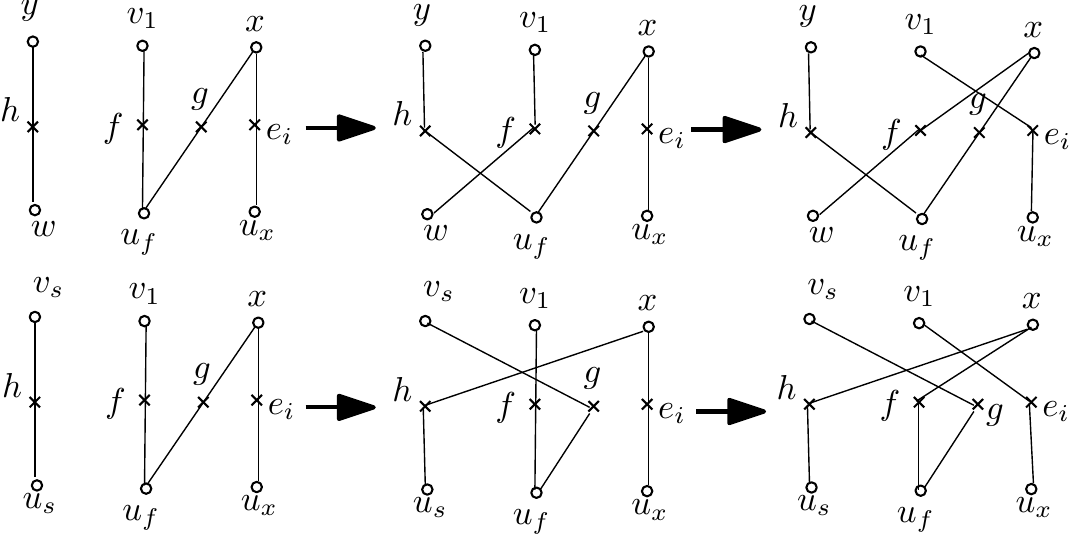}
\caption{A graphical representation of the situations discussed in Case (2b)}
\label{fig:simplecase2b}
\end{figure}

Therefore, given any node, a precise set of edges that it should neighbor, and a set of vertices that are the endpoints of those edges, we can use half-edge-rewirings to transform any graph $G$ to $G'$ that has this property, provided the set of edges is graphical. \qed

\end{proof}

Now that we have shown that both $\cal A$ and $\cal B$ converge to the uniform distribution over their respective state spaces, the next question is how quickly this happens. Note that from the proof that the state space of $\cal B$ is connected, we can upperbound the diameter of the state space by $3m$. The diameter provides a lower bound on the mixing time. In the next section, we will empirically estimate the mixing time to be also linear in $m$.

\section{Estimating the Mixing Time of the Markov Chain}

The Markov chain $\cal A$ is very similar to one analyzed by Kannan, Tetali and Vempala~\cite{ktv}. We can exactly use their canonical paths and analysis to show that the mixing time is polynomial. This result follows directly from Theorem 3 of~\cite{ktv} for chain $\cal A$. This is because the joint degree matrix configuration model can be viewed as $|\dv|$ complete, bipartite, and disjoint components. These components should remain disjoint, so the Markov Chain can be viewed as a `meta-chain' which samples a component and then runs one step of the Kannan, Tetali and Vempala chain on that component. Even though the mixing time for this chain is provably polynomial, this upper bound is too large to be useful in practice.

The analysis to bound the mixing time for $\cal B$ chain is significantly more complicated. One approach is to use the canonical path method to bound the congestion of this chain. The standard trick is to define a path from $G_1$ to $G_2$ that fixes the misplaced edges identified by $G_1 \oplus G_2$ in a globally ordered way. However, this is difficult to apply to chain $\cal B$ because fixing a specific edge may not be atomic, i.e. from the proof of Theorem~\ref{thm:swaps} it may take up to 4 swaps to correctly connect a vertex with an endpoint if there are conflicts with the other degree neighborhoods. These swaps take place in other degree neighborhoods and are not local moves. Therefore, this introduces new errors that must be fixed, but can not be incorporated into $G_1 \oplus G_2$. In addition, step (4) also prevents us from using path coupling as a proof of the mixing time.

Given that bounding the mixing time of this chain seems to be difficult without new techniques or ideas, we use a series of experiments that substitute the $\emph{autocorrelation time}$ for the mixing time. 

\subsection{Autocorrelation Time}

Autocorrelation time is a quantity that is related to the mixing time and is popular among physicists. We will give a brief introduction to this concept, and refer the reader to Sokal's lecture notes for further details and discussion~\cite{sokal}.

The autocorrelation of a signal is the cross-correlation of the signal with itself given a lag $t$. More formally, given a series of data $\langle X_i \rangle$ where each $X_i$ is a drawn from the same distribution $X$ with mean $\mu$ and variance $\sigma$, the autocorrelation function is $R_X(t) = \frac{E[(X_i-\mu)(X_{i-t}-\mu)]}{\sigma^2}$.

Intuitively, the inherent problem with using a Markov Chain sampling method is that successive states generated by the chain may be highly correlated. If we were able to draw independent samples from the stationary distribution, then the autocorrelation of that set of samples with itself would go to 0 as the number of samples increased. The autocorrelation time is capturing the size of the gaps between sampled states of the chain needed before the autocorrelation of this `thinned' chain is very small. If the thinned chain has 0 autocorrelation, then it must be exactly sampled from the stationary distribution. In practice, when estimating the autocorrelation from a finite number of samples, we do not expect it to go to exactly 0, but we do expect it to `die away' as the number of samples and gap increases.

\begin{definition}
The exponential autocorrelation time is $\tau_{exp,X} = \limsup_{t\rightarrow \infty}\frac{t}{-\log|R_X(t)|}$~\cite{sokal}.

\end{definition}

\begin{definition}
The integrated autocorrelation time is $\tau_{int,X} = \frac 1 2 \sum_{t = -\infty}^{\infty}R_X(t) = \frac 1 2 + \sum_{t=1}^{\infty}R_X(t)$~\cite{sokal}.
\end{definition}

The difference between the exponential autocorrelation time and the integrated autocorrelation time is that the exponential autocorrelation time measures the time it takes for the chain to reach equilibrium after a cold start, or `burn-in' time. The integrated autocorrelation time is related to the increase in the variance over the samples from the Markov Chain as opposed to samples that are truly independent. Often, these measurements are the same, although this is not necessarily true. 

We can substitute the autocorrelation time for the mixing time because they are, in effect, measuring the same thing - the number of iterations that the Markov Chain needs to run for before the difference between the current distribution and the stationary distribution is small. We will use the integrated autocorrelation time estimate.

\subsection{Experimental Design}
We used the Markov Chain $\cal B$ in two different ways. First, for each of the smaller datasets, we ran the chain for 50,000 iterations 15 times. We used this to calculate the the autocorrelation values for each edge for each lag between 100 and 15,000 in multiples of 100. From this, we calculated the estimated integrated autocorrelation time, as well as the iteration time for the autocorrelation of each edge to drop under a threshold of 0.001. This is discussed in Section~\ref{sec:threshold}.

We also replicated the experimental design of Raftery and Lewis~\cite{Raftery95thenumber}. Given our estimates of the autocorrelation time for each size graph in Section~\ref{sec:threshold}, we ran the chain again for long enough to capture 10,000 samples where each sample had $x$ iterations of the chain between them. $x$ was chosen to vary from much smaller than the estimated autocorrelation time, to much larger. From these samples, we calculated the sample mean for each edge, and compared it with the actual mean from the joint degree matrix. We looked at the total variational distance between the sample means and actual means and showed that the difference appears to be converging to 0. We chose the mean as an evaluation metric because we were able to calculate the true means theoretically. We are unaware of another similarly simple metric.

We used the formulas for empirical evaluation of mixing time from page 14 of Sokal's survey~\cite{sokal}. In particular, we used the following:
\begin{itemize}
\item The sample mean is $\overline{\mu} = \frac 1 n \sum_{i=1}^n x_i$.
\item The sample unnormalized autocorrelation function is $\hat{C}(t) = \frac 1 {n-t} \sum_{i=1}^{n-t} (x_i-\overline{\mu})(x_{i+t}-\overline{\mu})$.
\item The natural estimator of $R_X(t)$ is $\hat{\rho}(t)=\hat{C}(t)/\hat{C}(0)$
\item The estimator for $\tau_{int,X}$ is $\hat{\tau}_{int}=\frac 1 2 \sum_{t=-(n-1)}^{n-1}\lambda(t)\hat{\rho}(t)$ where $\lambda$ is a `suitable' cutoff function.
\end{itemize}

For a sequence of length $x$, calculating the autocorrelation of gap $t$ requires $(x-t)^2$ dot products. Our experiments require that we calculate the autocorrelation for each possible edge in a graph for many lags. Thus running the full set of experiments requires $O(|V|^2x\log x)$ time and is prohibitive when $V$ is large. Note that $x$ must necessarily be at least $\Theta(E)$ as well, since the mixing time can not be sub-linear in the number of edges.  In Section~\ref{sec:edgevtime} we will discuss results on the smaller datasets (AdjNoun, Dolphins, Football, Karate, and LesMis) that suggest a more feasible method for estimating autocorrelation time for larger graphs. We use this method to evaluate the autocorrelation time for the larger graphs as well, and present all of the results together. Rather than running the chain for 15,000 steps for the larger graphs, we selected more appropriate stopping conditions that were generally $10|E|$ based on the results for smaller graphs.

\paragraph*{Data Sets}

We have used several publicly available datasets, Word Adjacencies~\cite{adjnoun}, Les Miserables~\cite{lesmis}, American College Football~\cite{football}, the Karate Club~\cite{karate}, the Dolphin Social Network~\cite{dolphins}, C. Elegans Neural Network (celegans)~\cite{celegans1, celegans2}, Power grid (power)~\cite{power}, Astrophysics collaborations (astro-ph)~\cite{astro-ph}, High-Energy Theory collaborations (hep-th)~\cite{hep-th}, Coauthorships in network science (netscience)~\cite{netscience}, and a snapshot of the Internet from 2006 (as-22july)~\cite{as-22july}. In the following $|V|$ is the number of nodes, $|E|$ is the number of edges and $|\jdm|$ is the number of non-zero entries in the joint degree matrix.

\begin{table}[ht]
\begin{center}
\begin{tabular}{|c|c|c|c|}
\hline
Dataset & $|E|$ &$|V|$ &$|\jdm|$\\
\hline
AdjNoun & 425 & 112 & 159\\
\hline
as-22july & 48,436 & 22,962 & 5,496\\
\hline
astro-ph & 121,251 & 16,705 & 11,360\\
\hline
celegans & 2,359 & 296 &642\\
\hline
Dolphins & 159 & 62 & 61\\
\hline
Football & 616 & 115 &18\\
\hline
hep-th & 15,751 & 8,360 & 629\\
\hline
Karate & 78 & 34 &40\\
\hline
LesMis & 254 & 77 &99\\
\hline
netscience & 2,742 & 1,588 & 184\\
\hline
power & 6,594 & 4,940 & 108\\
\hline
\end{tabular}
\end{center}
\caption{Details about the datasets, $|V|$ is the number of nodes, $|E|$ is the number of edges and $|\jdm|$ is the number of unique entries in the $\jdm$.}
\end{table}


\subsection{Relationship Between Mean of an Edge and Autocorrelation}\label{sec:edgevtime}

For each of the smaller graphs, AdjNoun, Dolphins, Football, Karate and LesMis, we ran the Markov Chain 10 times for 50,000 iterations and collected an indicator variable for each potential edge. For each of these edges, and each run, we calculated the autocorrelation function for values of $t$ between 100 and 15,000 in multiples of 100. For each edge, and each run, we looked at the $t$ value where the autocorrelation function first dropped below the threshold of $0.001$. We then plotted the mean of these values against the mean of the edge, i.e. if it connects vertices of degree $d_i$ and $d_j$ (where $d_i \neq d_j$) then $\mu_e = \jdm_{d_i,d_j}/d_id_j$ or $\mu_e = \jdm_{d_i,d_i}/{d_i \choose 2}$ otherwise. The three most useful plots are given in Figures~\ref{fig:adjnoun_mean} and~\ref{fig:karate_mean} as the other graphs did not contain a large range of mean values.

\begin{figure}

\begin{minipage}[b]{0.5\linewidth}\centering
\centering\includegraphics[width=\columnwidth]{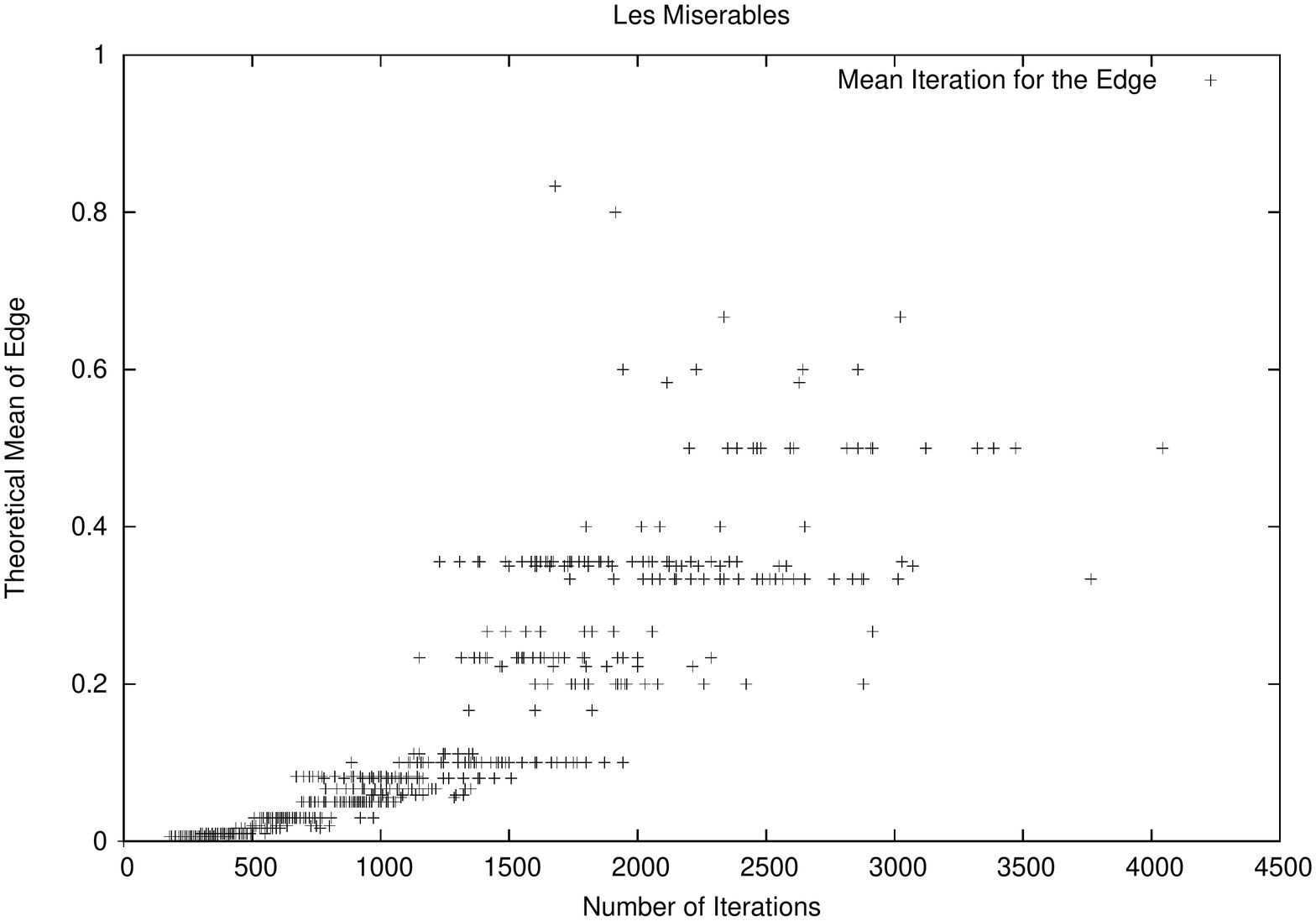}
\end{minipage}\hfill
\begin{minipage}[b]{0.5\linewidth}\centering
\centering\includegraphics[width=\columnwidth]{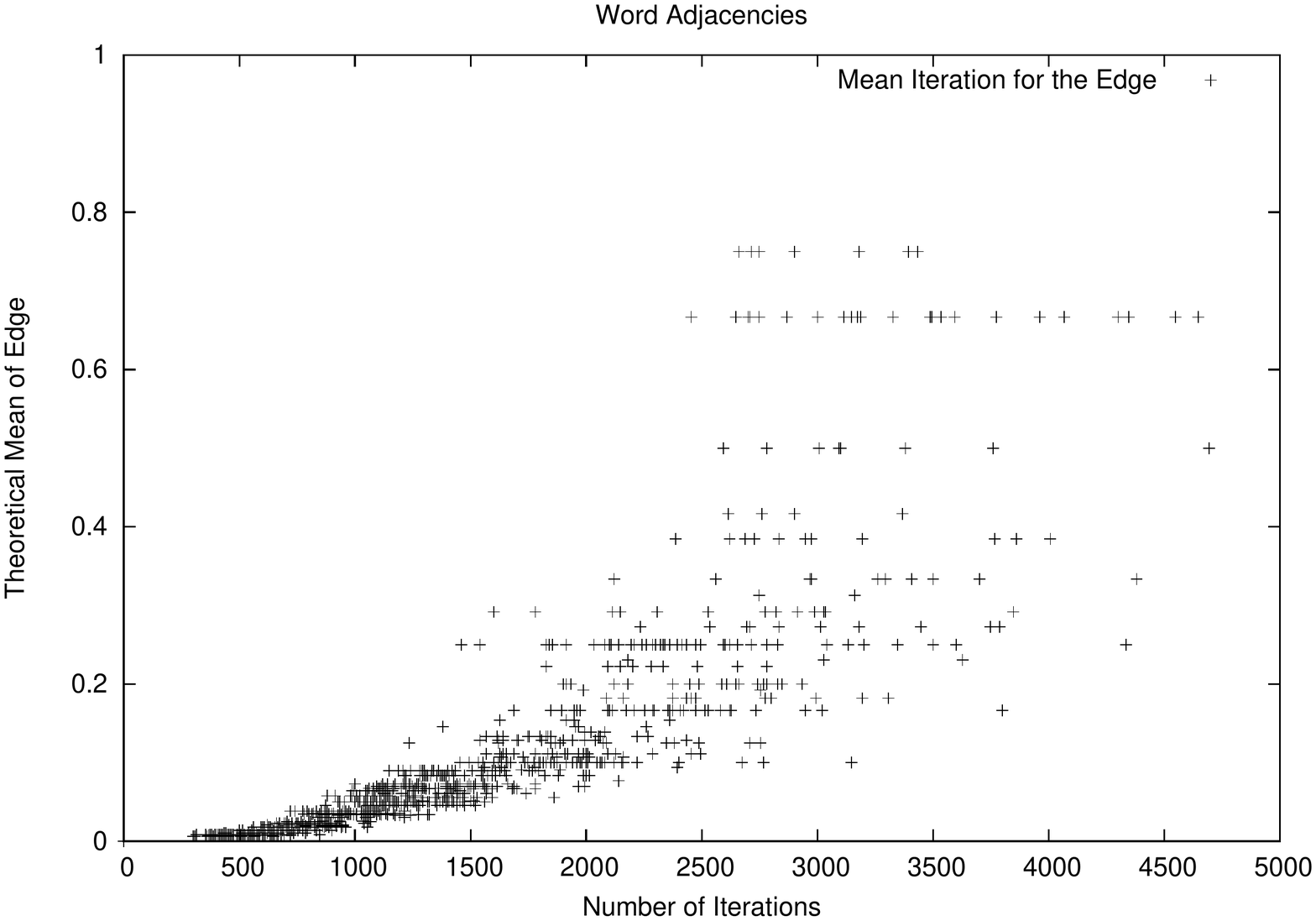}

\end{minipage}\hfill
\caption{The time for an edge's estimated autocorrelation function to pass under the threshold of 0.001 versus $\mu_e$ for that edge for LesMis and AdjNoun from L to R.}
\label{fig:adjnoun_mean}
\end{figure}

From these results, we identified a potential relationship between $\mu_e$ and the time to pass under a threshold. Unfortunately, none of our datasets contained a significant number of edges with larger $\mu_e$ values, i.e. between 0.5 and 1. In order to test this hypothesis, we designed a synthetic dataset that contained the many edges with values of $\mu_e$ at $\frac i {20}$ for $i=1,\cdots 20$. We describe the creation of this dataset in the appendix.

The final dataset we created had 326 edges, 194 vertices and 21 distinct $\jdm$ entries. We ran the Markov Chain 200 times for this synthetic graph. For each run, we calculated the threshold value for each edge. 
Figure~\ref{fig:karate_mean} shows the edges' mean vs its mean time for the autocorrelation value to pass under 0.001. We see that there is a roughly symmetric curve that obtains its maximum at $\mu_e=0.5$. 

\begin{figure}[!h]
\begin{minipage}[b]{0.5\linewidth}\centering
\centering\includegraphics[width=\columnwidth]{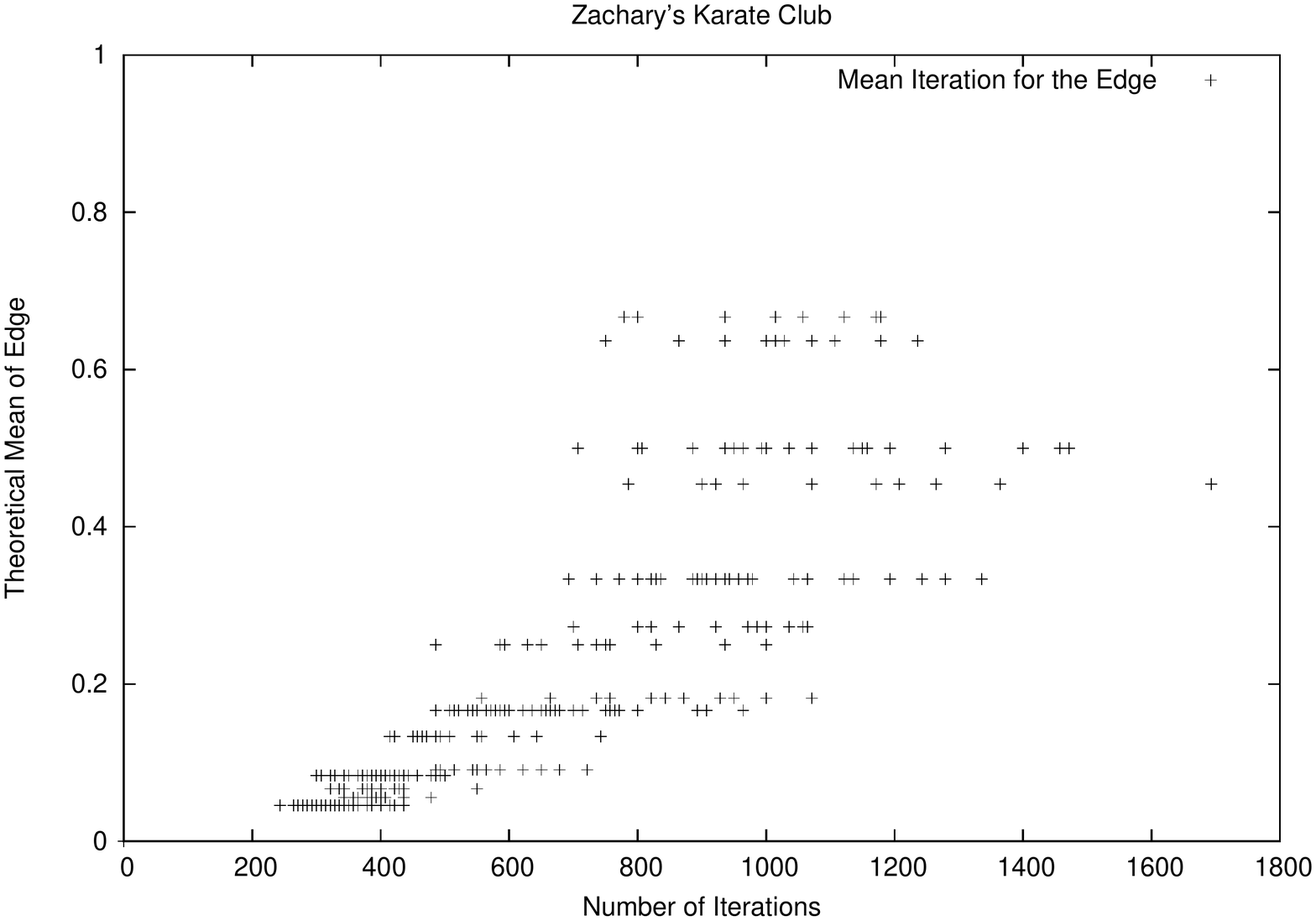}

\end{minipage}\hfill
\begin{minipage}[b]{0.5\linewidth}\centering
\centering\includegraphics[width=\columnwidth]{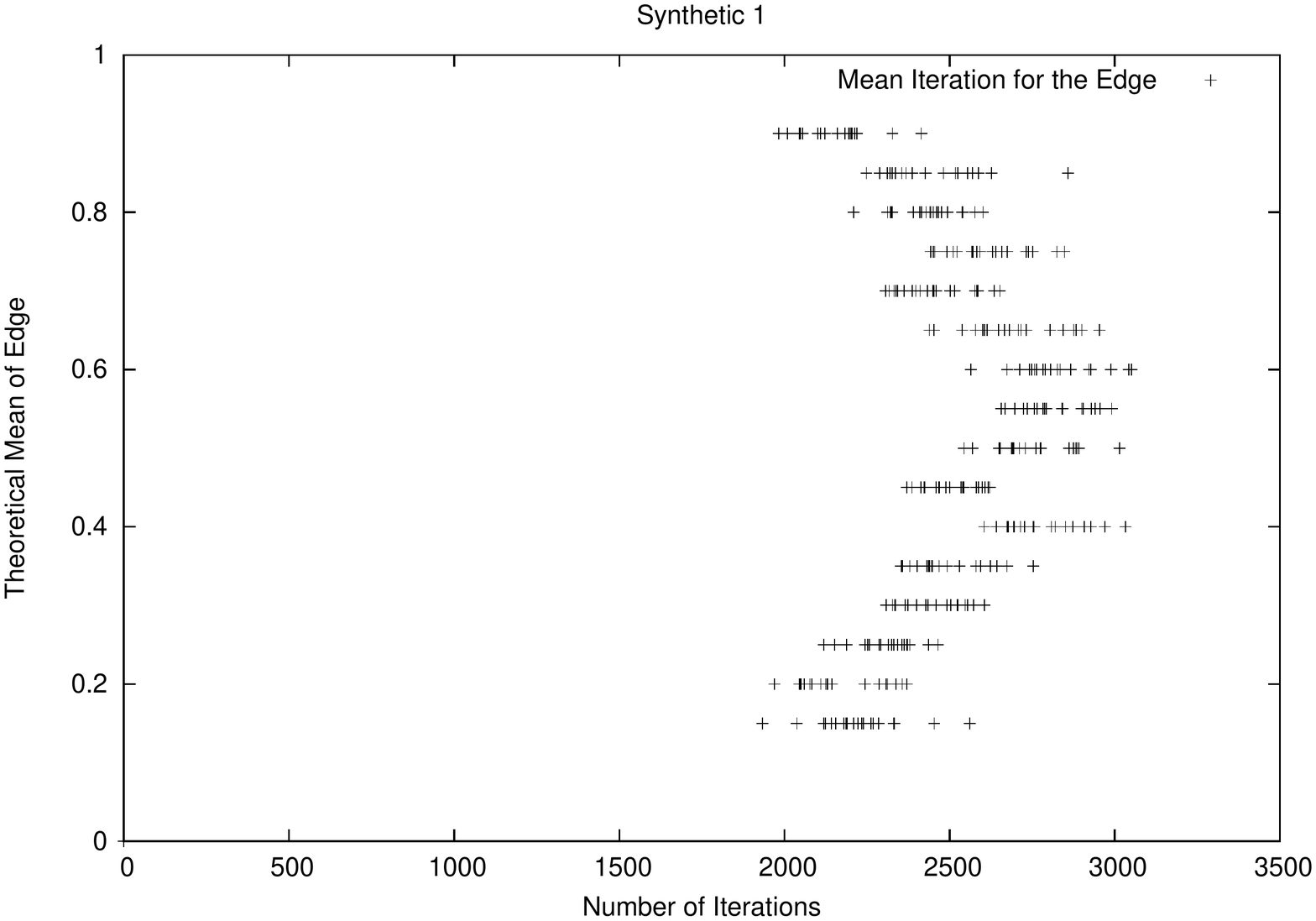}
\end{minipage}\hfill
\caption{The time for an edge's estimated autocorrelation function to pass under the threshold of 0.001 versus $\mu_e$ for that edge for Karate and the synthetic dataset. The synthetic dataset has a larger range of $\mu_e$ values than the real datasets and a significant number of edges for each value.}
\label{fig:karate_mean}
\end{figure}

This result suggests a way to estimate the autocorrelation time for larger graphs without repeating the entire experiment for every edge that could possibly appear. One can calculate $\mu_e$ for each edge from the JDM and sample edges with $\mu_e$ around 0.5. We use this method for selecting our subset of edges to analyze. In particular, we sampled about 300 edges from each of the larger graphs. For all of these except for power, the $\mu_e$ values were between 0.4 and 0.6. For power, the maximum $\mu_e$ value is about 0.15, so we selected edges with the largest $\mu$ values.

\subsection{Autocorrelation Values}\label{sec:threshold}
For each dataset and each run we calculated the unnormalized autocorrelation values. For the smaller graphs, this entailed setting $t$ to every value between 100 and 15,000 in multiples of 100. We randomly selected 1 run for each dataset and graphed the autocorrelation values for each of the edges. We present the data for the Karate and Dolphins datasets in Figures~\ref{fig:karateac} and~\ref{fig:dolphinsac}. For the larger graphs, we changed the starting and ending points, based on the graph size. For example, for Netscience was analyzed from 2,000 to 15,000 in multiples of 100, while as-22july was analyzed from 1,000 to 500,000 in multiples of 1,000.

\begin{figure}
\begin{minipage}[b]{0.45\linewidth}\centering
\includegraphics[width = \columnwidth,viewport=90 240 530 505, clip]{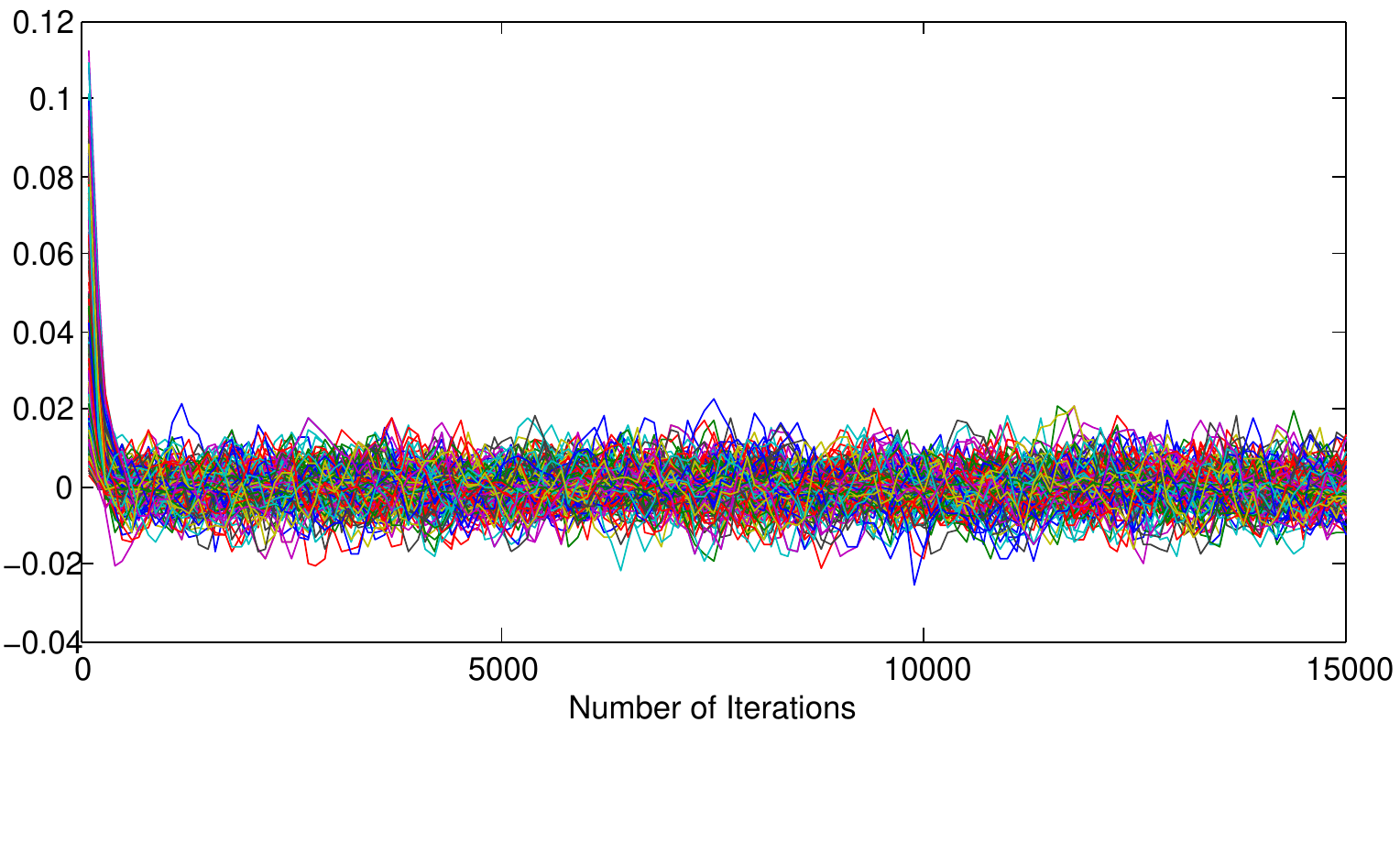}
\caption{The exponential drop-off for Karate appears to end after 400 iterations.}
\label{fig:karateac}
\end{minipage}\hfill%
\begin{minipage}[b]{0.45\linewidth}\centering
\includegraphics[width = \columnwidth,viewport=90 240 530 505, clip]{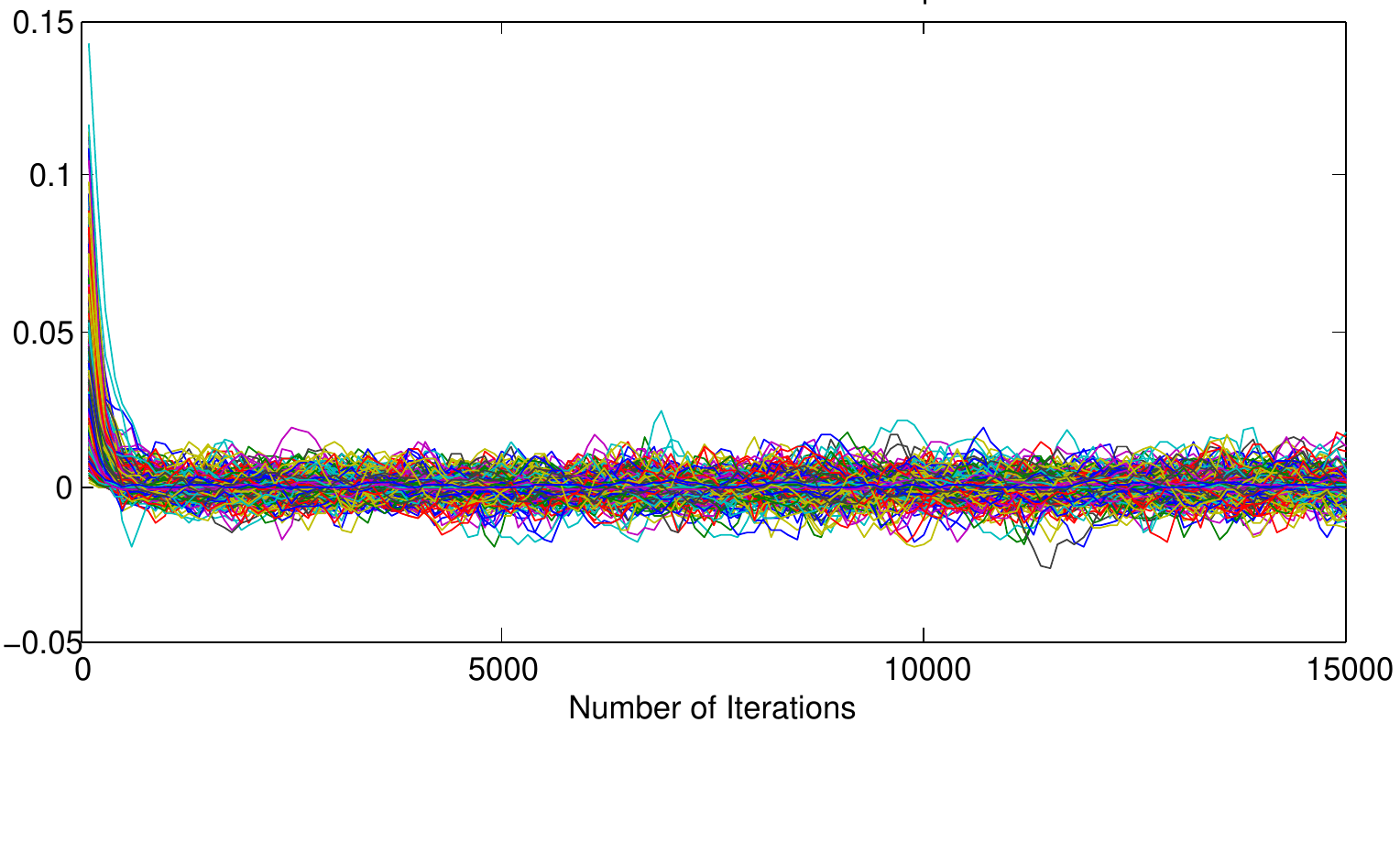}
\caption{The exponential drop-off for Dolphins appears to end after 600 iterations.}
\label{fig:dolphinsac}
\end{minipage}
\end{figure}


All of the graphs exhibit the same behavior. We see an exponential drop off initially, and then the autocorrelation values oscillate around 0. This behavior is due to the limited number of samples, and a bias due to using the sample mean for each edge. If we ignore the noisy tail, then we estimate that the autocorrelation `dies off' at the point where the mean absolute value of the autocorrelation approximately converges, then we can locate the `elbow' in the graphs. This estimate for all graphs is given in Table~\ref{table:summary} at the end of this Section.

\subsection{Estimated Integrated Autocorrelation Time}

For each dataset and run, we calculated the estimated integrated autocorrelation time. For the datasets with fewer than 1,000 edges, we calculated the autocorrelation in lags of 100 from 100 to 15,000 for each dataset. For the larger ones, we used intervals that depended on the total size of the graph. We estimate $\hat{\rho}(t)$ as the size of the intervals times the sum of the values. The cut-off function we used for the smaller graphs was $\lambda(t)=1$ if $0<t<15,000$ and 0 otherwise. This value was calculated for each edge. In Table~\ref{table:estintauto} we present the mean, maximum and minimum estimated integrated autocorrelation time for each dataset over the runs of the Markov Chain using three different methods. For each of the edges, we first calculated the mean, median and max estimated integrated autocorrelation value over the various runs. Then, for each of these three values for each edge, we calculated the max, mean and min over all edges.  For each of the graphs, the data series representing the median and max have each had their x-values perturbed slightly for clarity.

\begin{table}
\begin{tabular}{|c|c|ccc|ccc|ccc|}
\hline
Dataset & $|E|$ & mean & max & min & median & max & min & maximum & max & min\\
\hline
Karate & 78 & 288.92 & 444.1 & 221.13 & 288.31 & 443 & 217.63 & 382.59 & 608.06 & 268.95\\
Dolphins & 159 & 383.21 & 553.84 & 256.13 & 377.4 & 550.99 & 211.44 & 528.86 & 1134.1 & 397.35\\
LesMis & 254 & 559.77 & 931.35 & 129.45 & 542.43 & 895.57 & 57.492 & 894.08 & 2598.6 & 342.76\\
AdjNoun & 425 & 688.71 & 1154.9 & 156.49 & 659.06 & 1160.3 & 66.851 & 1186.1 & 4083.6 & 350.97\\
Football & 616 & 962.42 & 2016.9 & 404.77 & 925.97 & 1646.9 & 349.12 & 1546.4 & 7514.3 & 967\\
celegans &2359 & 3340.2 & 4851.4 & 2458.8 & 3235.7 & 4861.4 & 2323.6 & 4844.6 & 7836.9 & 3065.5\\
netscience & 2742 & 1791.4 & 3147.2 & 1087.7 & 1658.3 & 3033.2 & 937.8382 & 3401 & 7404 & 1894.7\\
power & 6594 & 6624.5 & 17933 & 2166.9 & 4768.8 & 16901 & 250.6012 & 20599 & 54814 & 7074.7\\
hep-th & 15751 & 26552 & 36816 & 14976 & 25608 & 37004 & 14130 & 46309 & 64936 & 25753\\
as-22july & 48436 & 89637 & 139280 & 60627 & 87190 & 152490 & 58493 & 121930 & 256520 & 76214\\
astro-ph & 121251 & 121860 & 298970 & 37706 & 119900 & 321730 & 46830 & 152930 & 408000 & 84498\\
\hline
\end{tabular}
\caption{A summary of all the estimated integrated autocorrelation times. Mean refers to taking the mean autocorrelation time for each edge, and then the mean, min and max of these values over all measured edges. Similarly, median is the median value for each edge, while max is the maximum for each edge.}
\label{table:estintauto}
\end{table}

These values are graphed on a log-log scale plot. Further, we also present a graph showing the ratio of these values to the number of edges. The ratio plot, Figure~\ref{fig:intratio}, suggests that the autocorrelation time may be a linear function of the number of edges in the graph, however the estimates are noisy due to the limited number of runs.

\begin{figure}
\includegraphics[width=80mm,  viewport = 50 50 300 290]{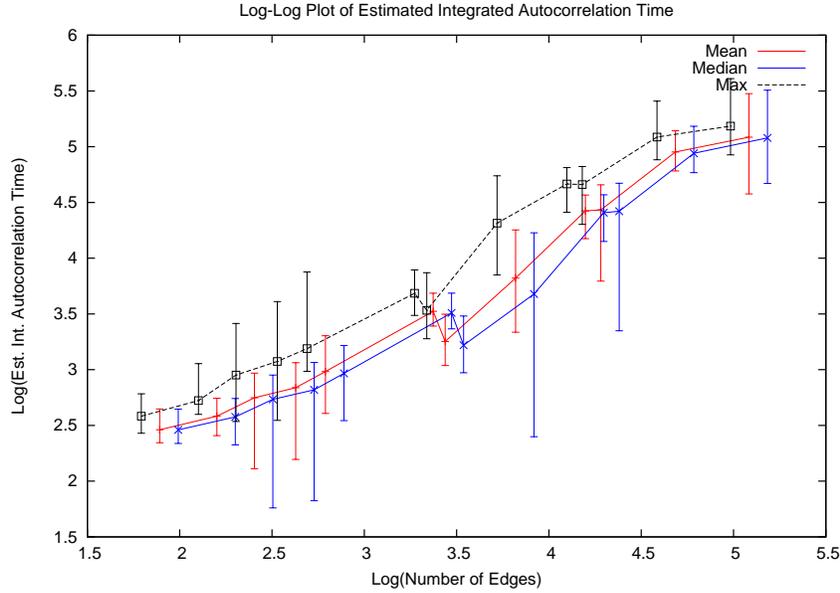}
\caption{The max, median and min values over the edges for the est. int. autocorrelation times in a log-log plot. L to R in order of size: Karate, Dolphins, LesMis, AdjNoun, Football, celegans, netscience, power, hep-th, as-22july and astro-ph}
\label{fig:intauto}
\end{figure}

\begin{figure}
\includegraphics[width = 80mm, viewport = 50 50 300 290]{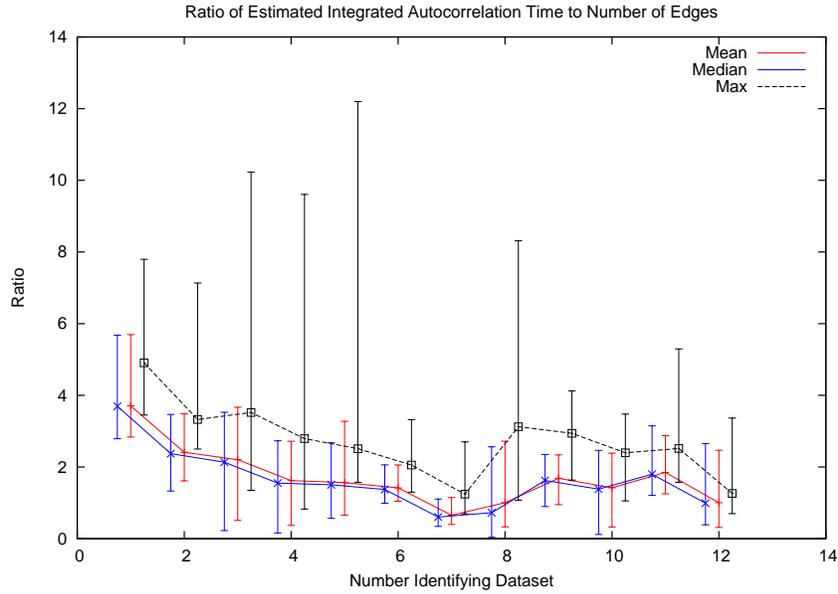}
\caption{The ratio of the max, median and min values over the edges to the number of edges for the estimated integrated autocorrelation times. L to R in order of size: Karate, Dolphins, LesMis, AdjNoun, Football, celegans, netscience, power, hep-th, as-22july and astro-ph}
\label{fig:intratio}
\end{figure}

All three metrics give roughly the same picture. We note that there is much higher variance in estimated autocorrelation time for the larger graphs. If we consider the evidence of the log-log plot and the ratio plot, we suspect that the autocorrelation time of this Markov Chain is linear in the number of edges.


\subsection{The Sample Mean Approaches the Real Mean for Each Edge}

Given the results of the previous experiment estimating the integrated autocorrelation time, we next executed an experiment suggested by Raftery and Lewis~\cite{Raftery95thenumber}. First we note that for each edge $e$, we know the true value of $P(e\in G| G\mbox{ has } \jdm)$ is exactly $\frac{\jdm_{k,l}}{\dv_k\dv_l}$ or $\frac{\jdm_{k,k}}{{\dv_k \choose 2}}$ if $e$ is an edge between degrees $k$ and $l$. This is because there are $\dv_k\dv_l$ potential $(k,l)$ edges that show up in any graph with a fixed $\jdm$, and each graph has $\jdm_{k,l}$ of them. If we consider the graphs as being labeled, then we can see that each edge has an equal probability of showing up when we consider permutations of the orderings.

Thus, our experiment was to take samples at varying intervals, and consider how the sample mean of each edge compared with our known theoretical mean. For the smaller graphs, we took 10,000 samples at varying gaps depending on our estimated integrated autocorrelation time and repeated this 10 times. 
Additionally, we saw that the total variational distance quickly converged to a small, but non-zero value. We repeated this experiment with 20,000 samples and, for the two smallest graphs, Karate and Dolphins, we repeated the experiment with 5,000 and 40,000 samples. These results show that this error is due to the number of samples and not the sampler. For the graphs with more than 1,000 edges, each run resulted in 20,000 samples at varying gaps, and this was repeated 5 times. We present these results in Figures 18 through 28. If $S_{e,g}$ is the sample mean for edge $e$ and gap $g$, and $\mu_e$ is the true mean, then the graphed value is $\sum_{e} |S_{e,g}-\mu_e|/ \sum_e \mu_e$.

\begin{figure}
\begin{minipage}[b]{0.48\linewidth}\centering
\includegraphics[width=\columnwidth]{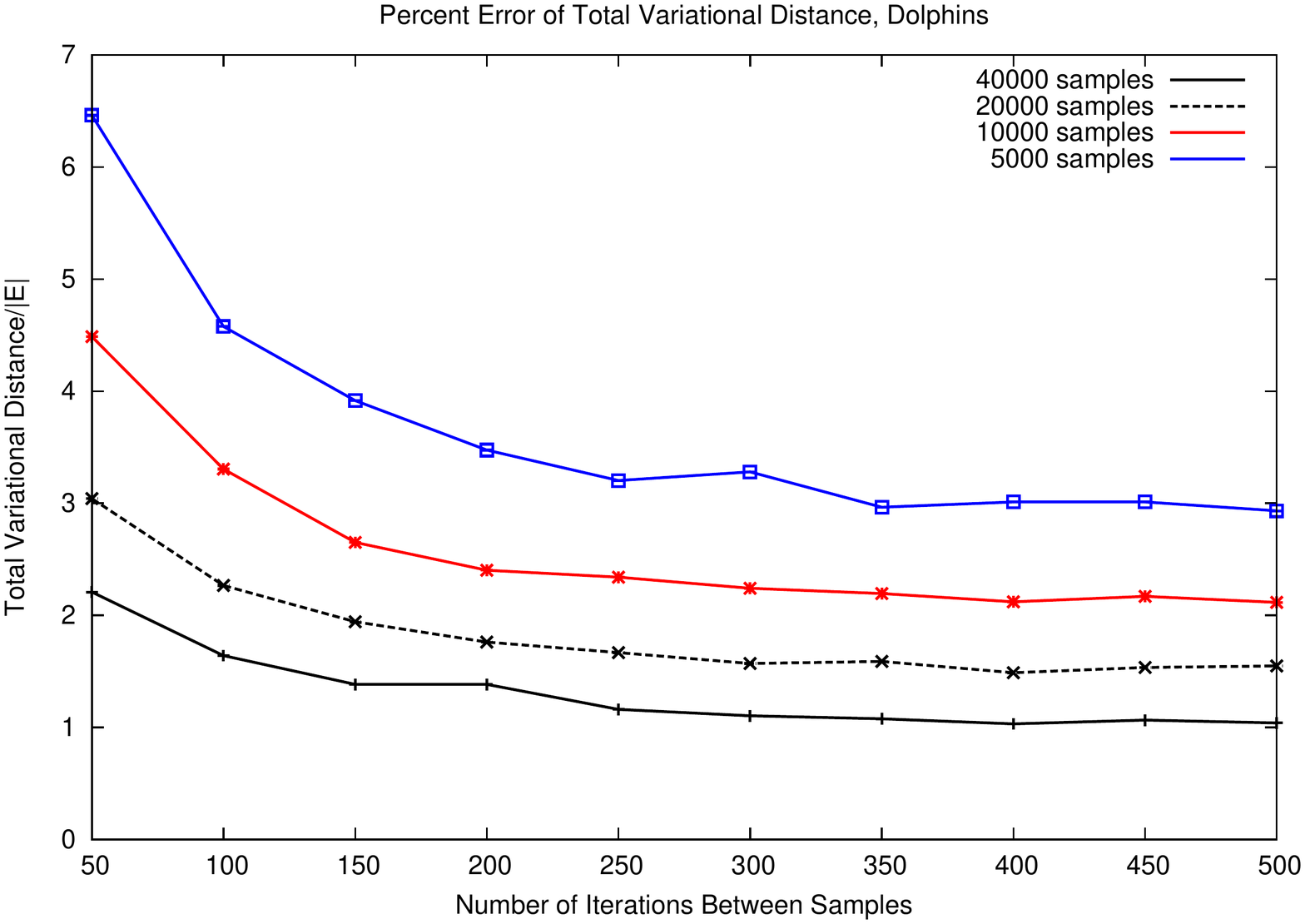} 
\caption{The Dolphin Dataset with 5,000 to 40,000 samples}\label{fig:dolphinsallmeans}
\end{minipage}\hfill%
\begin{minipage}[b]{0.48\linewidth}\centering
\includegraphics[width=\columnwidth]{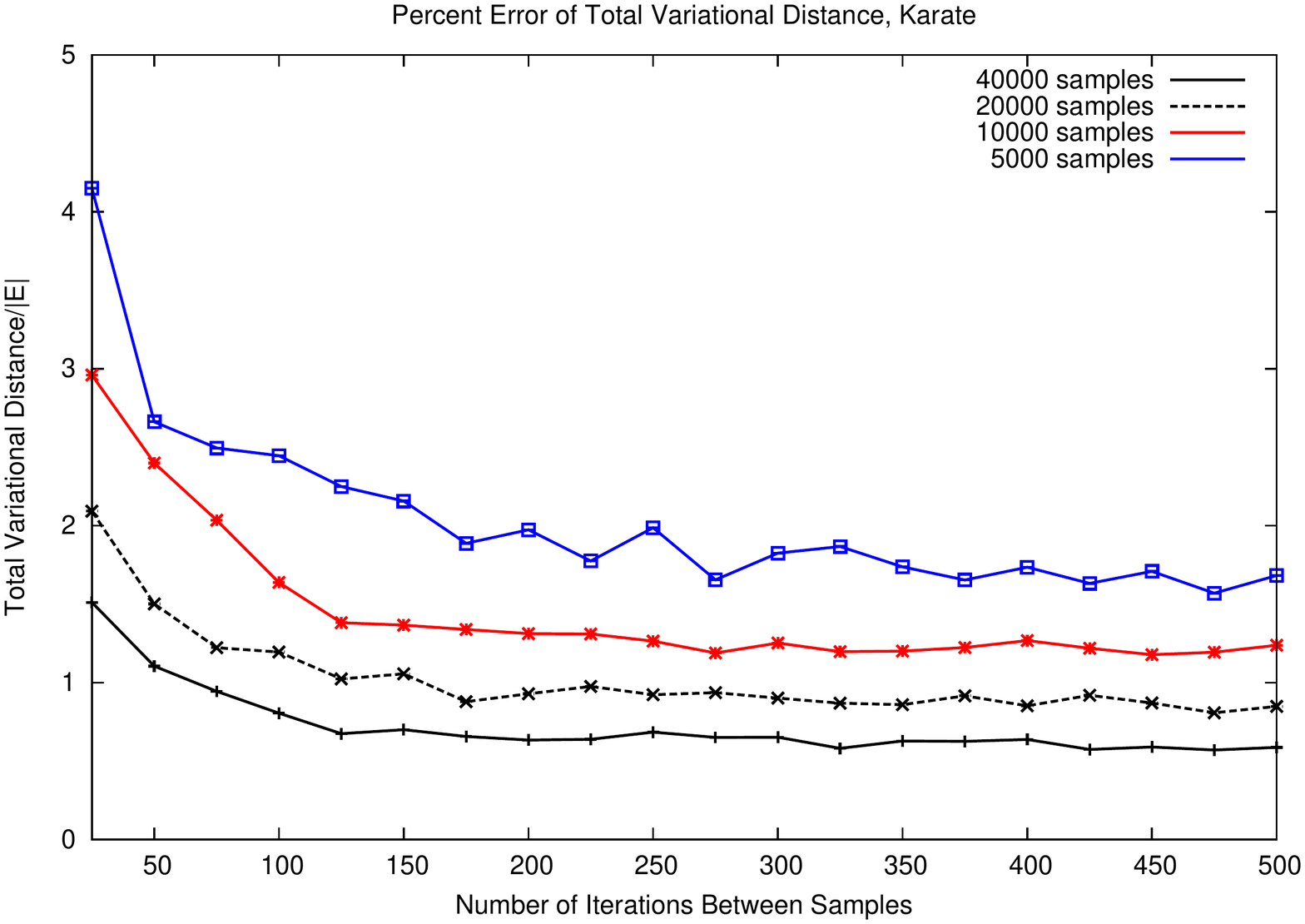}
\caption{The Karate Dataset with 5,000 to 40,000 samples}\label{fig:karateallmeans}
\end{minipage}
\end{figure}

\begin{figure}
\begin{minipage}[b]{0.48\linewidth}\centering
\includegraphics[width=\columnwidth]{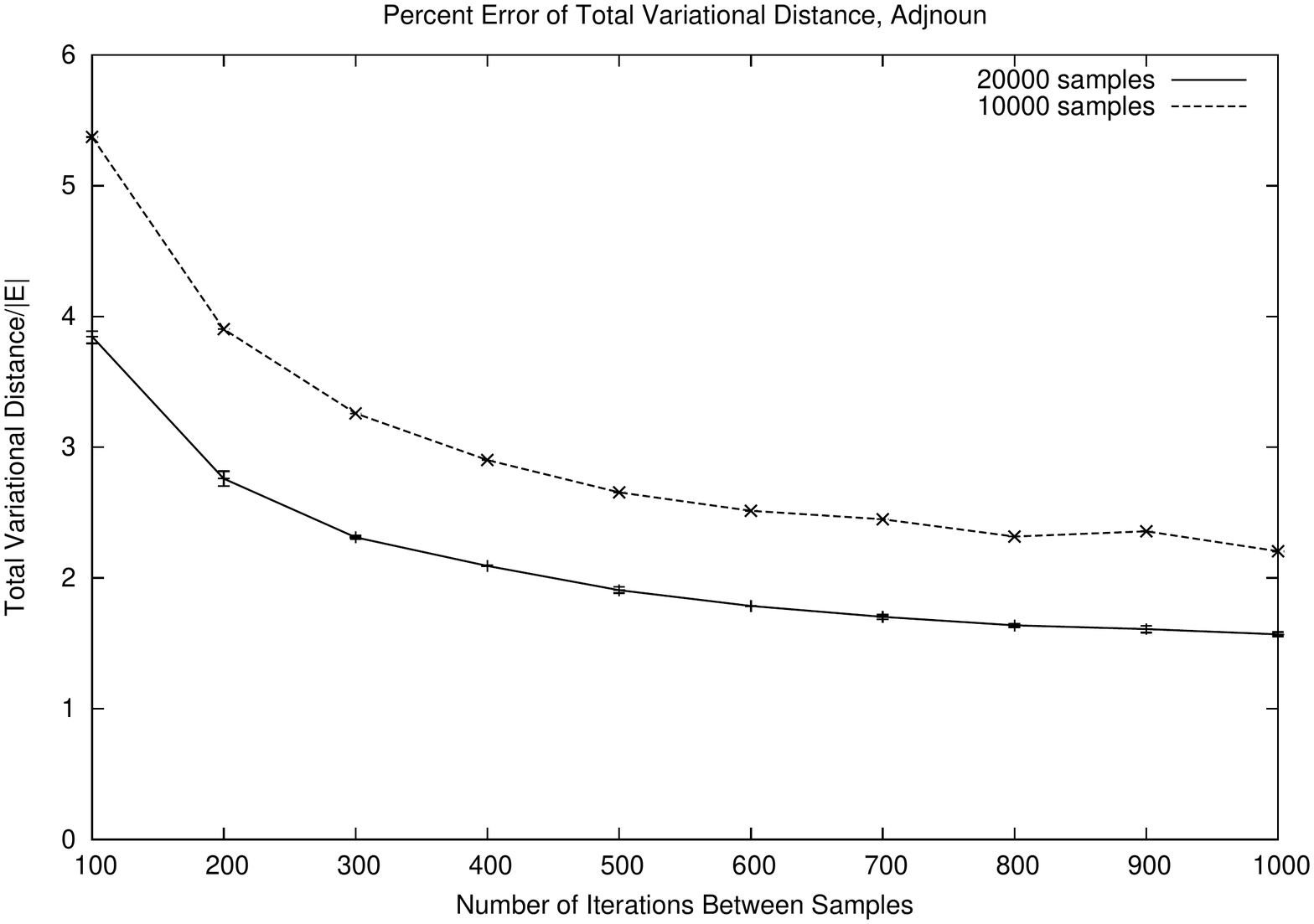} 
\caption{The AdjNoun Dataset with 10,000 and 20,000 samples}\label{fig:adjnounallmeans}
\end{minipage}\hfill%
\begin{minipage}[b]{0.48\linewidth}\centering
\includegraphics[width=\columnwidth]{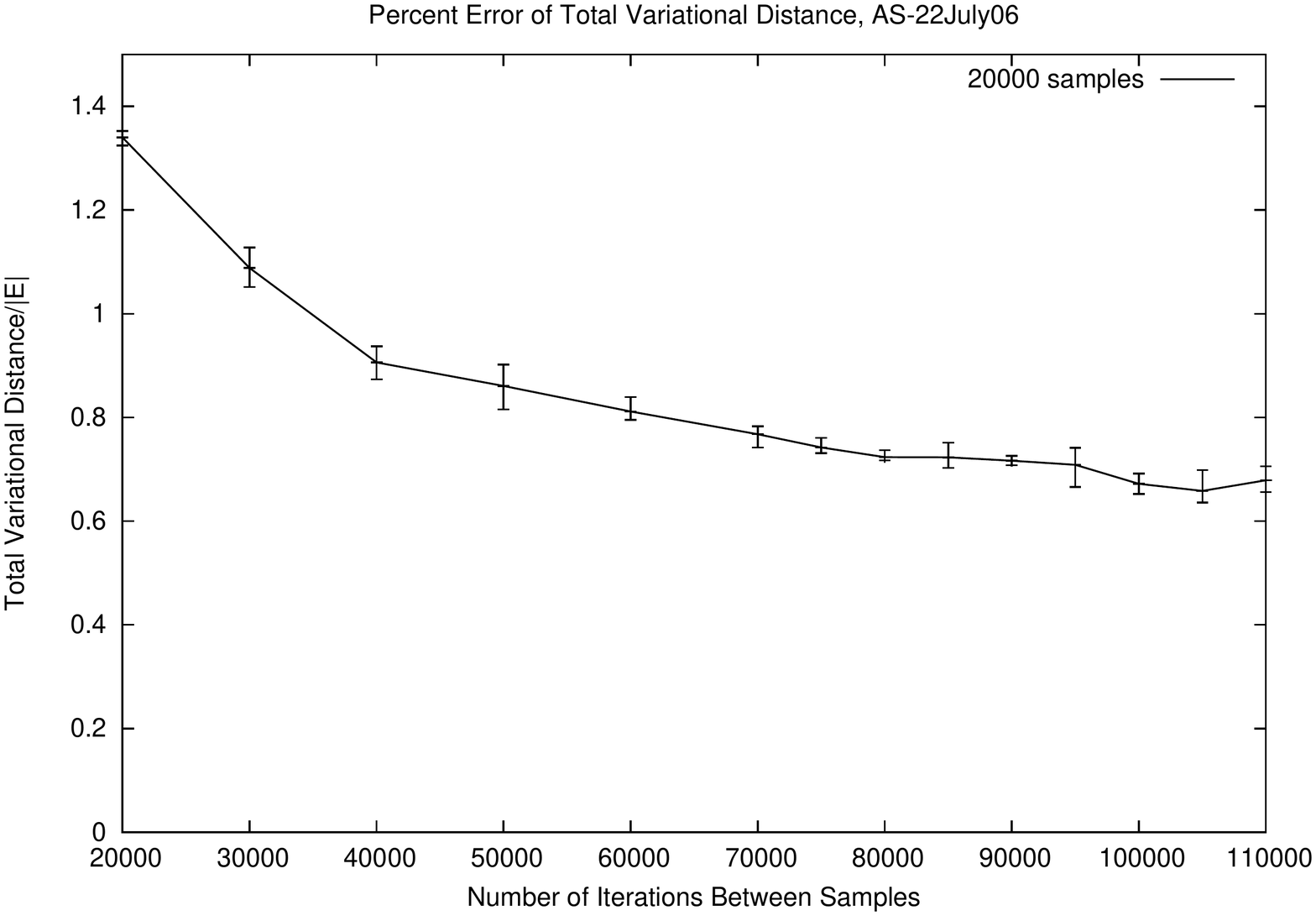}
\caption{The AS-22July06 Dataset with 20,000 samples}\label{fig:as-22julyallmeans}
\end{minipage}
\end{figure}

\begin{figure}
\begin{minipage}[b]{0.48\linewidth}\centering
\includegraphics[width=\columnwidth]{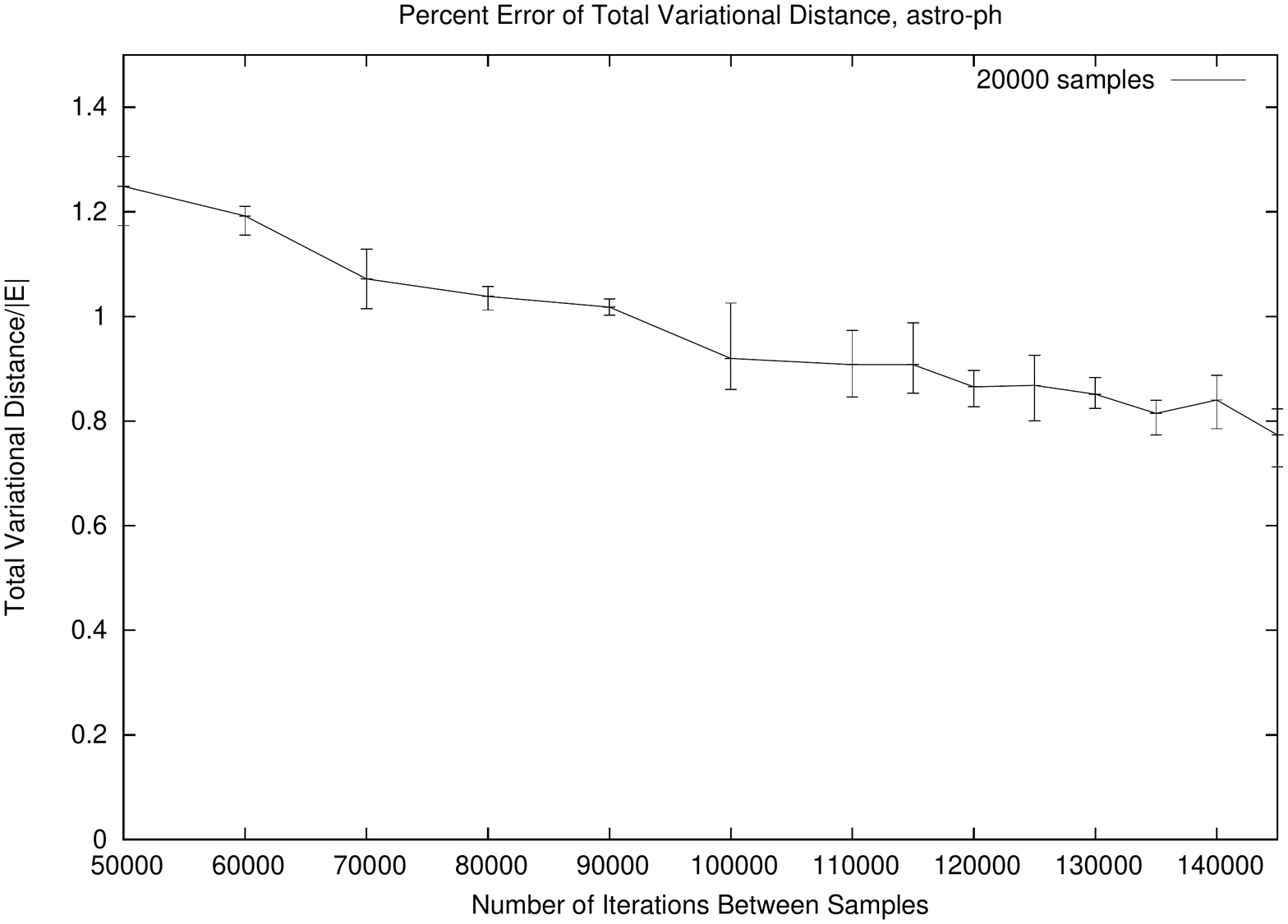} 
\caption{The Astro-PH Dataset with 20,000 samples}\label{fig:astroallmeans}
\end{minipage}\hfill%
\begin{minipage}[b]{0.48\linewidth}\centering
\includegraphics[width=\columnwidth]{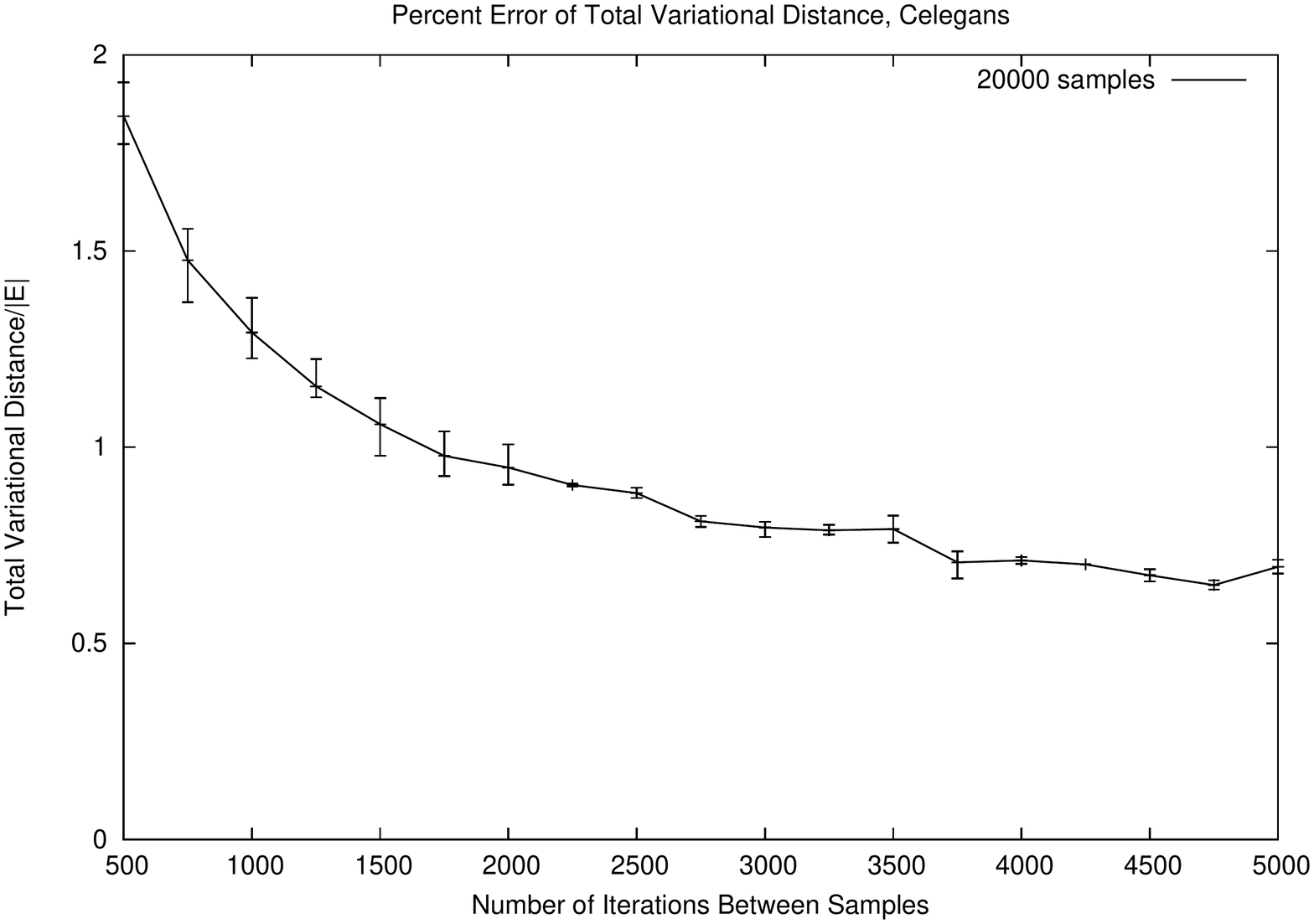}
\caption{The Celegans Dataset with 20,000 samples}\label{fig:celegansallmeans}
\end{minipage}
\end{figure}

\begin{figure}
\begin{minipage}[b]{0.48\linewidth}\centering
\includegraphics[width=\columnwidth]{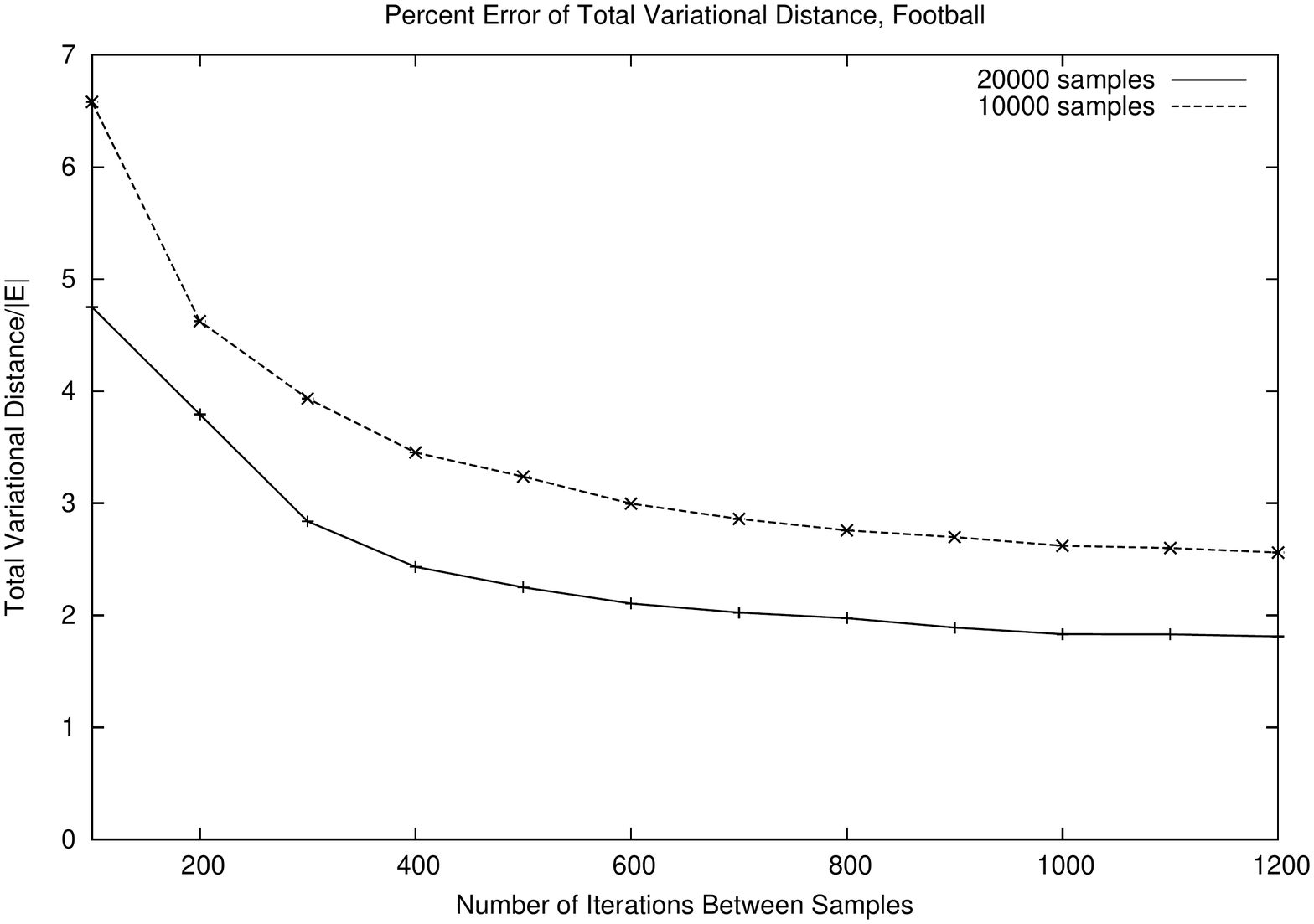} 
\caption{The Football Dataset with 10,000 and 20,000 samples}\label{fig:footballallmeans}
\end{minipage}\hfill%
\begin{minipage}[b]{0.48\linewidth}\centering
\includegraphics[width=\columnwidth]{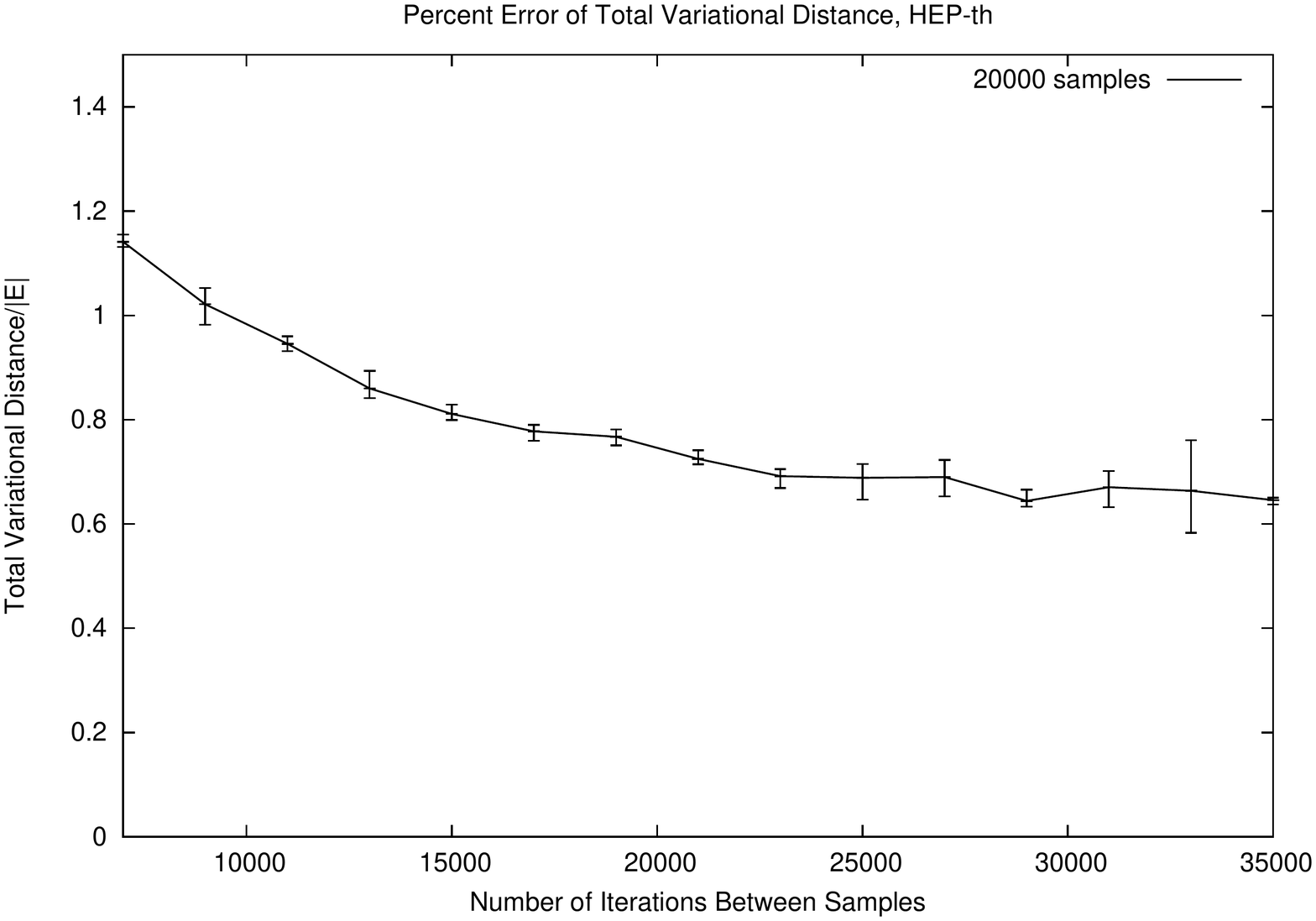}
\caption{The Hep-TH Dataset with 20,000 samples}\label{fig:hep-thallmeans}
\end{minipage}
\end{figure}

\begin{figure}
\begin{minipage}[b]{0.48\linewidth}\centering
\includegraphics[width=\columnwidth]{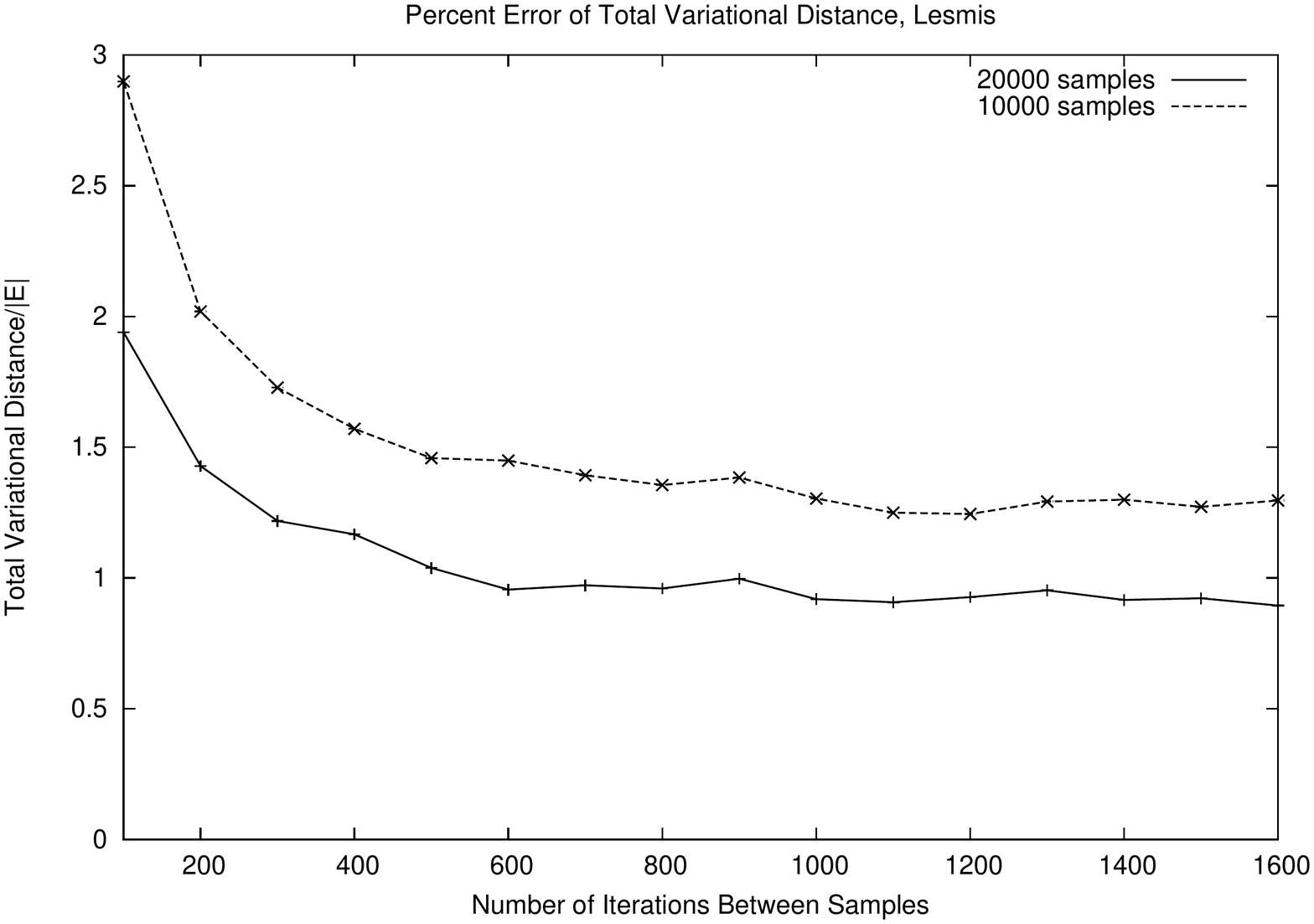} 
\caption{The LesMis Dataset with 10,000 and 20,000 samples}\label{fig:lesmisallmeans}
\end{minipage}\hfill%
\begin{minipage}[b]{0.48\linewidth}\centering
\includegraphics[width=\columnwidth]{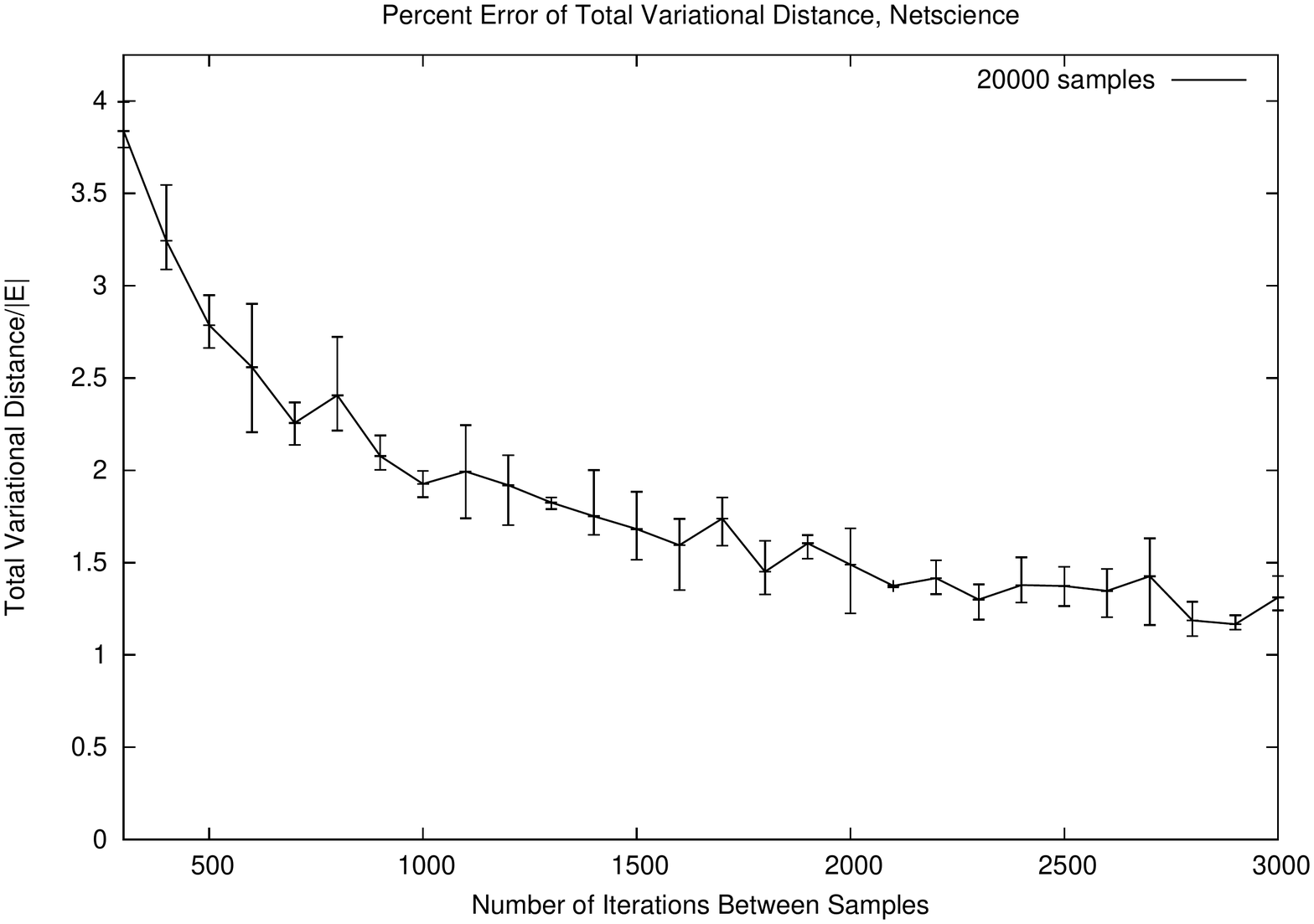}
\caption{The Netscience Dataset with 20,000 samples}\label{fig:netscienceallmeans}
\end{minipage}
\end{figure}

\begin{figure}
\begin{minipage}[b]{0.48\linewidth}\centering
\includegraphics[width=\columnwidth]{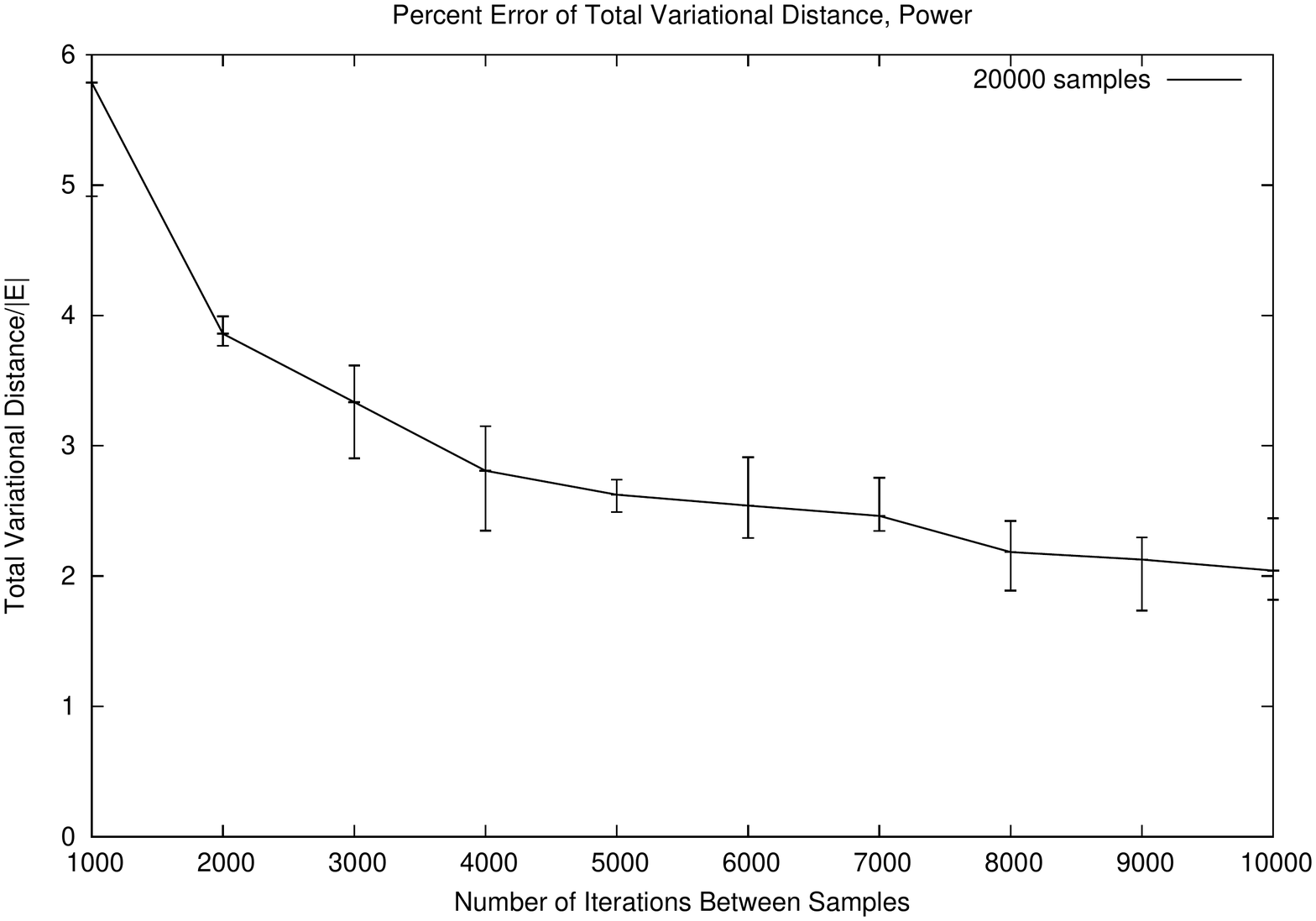} 
\caption{The Power Dataset with 20,000 samples}\label{fig:powerallmeans}
\end{minipage}\hfill%
\end{figure}

In all of the figures, the line runs through the median error for the runs and the error bars are the maximum and minimum values. We note that the maximum and minimum are very close to the median as they are within 0.05\% for most intervals. These graphs imply that we are sampling uniformly after a gap of 175 for the Karate graph. For the dolphin graph, we see very similar results, and note that the error becomes constant after a sampling gap of 400 iterations. 

For the larger graphs, we varied the gaps based on the graph size, and then focused on the area where the error appeared to be decreasing. Again, we see consistent results, although the residual error is higher. This is to be expected because there are more potential edges in these graphs, so we took relatively fewer samples per edge. A summary of the results can be found in Table~\ref{table:summary}.

\subsection{Summary of Experiments}
\begin{table}[ht]
\begin{center}
\begin{tabular}{|c|c|c|c|c|}
\hline
&$|E|$& Max EI& Mean Conv. & Thresh.\\
\hline
AdjNoun & 425 & 1186& 900 & 700\\
\hline
AS-22July& 48,436 &256,520 &95,000&156,744\\
\hline
Astro-PH &121,251 &408,000 &120,000&343,154\\
\hline
Celegans & 2,359 & 7836.9&3,750&7,691\\
\hline
Dolphins & 159 & 528& 400& 600\\
\hline
Football & 616 & 1546& 1000 & 900\\
\hline
Hep-TH & 15,751 & 64,936& 28,000 &22,397 \\
\hline
Karate & 78 & 382 & 175& 400\\
\hline
LesMis & 254 & 894& 800& 1000\\
\hline
Netscience & 2,742 & 7,404&2,000&7,017\\
\hline
Power & 6,594 & 54,814&8,000&7,270\\
\hline
\end{tabular}
\caption{A summary of estimates on convergence from the three experiments. The values are the Maximum Estimated Integrated Autocorrelation time (Max EI, the third column of Table 2), the Sample Mean Convergence iteration number, and the time to drop under the Autocorrelation Threshold. The Autocorrelation threshold was calculated as when the average absolute value of the autocorrelation was less than 0.0001}
\label{table:summary}
\end{center}
\end{table}

Based on the results in this table, our recommendation would be that running the Markov Chain for $5m$ steps would satisfy all running time estimates except for Power's results for the Maximum Estimated Integrated Autocorrelation time. This estimate is significantly lower than the result for Chain $\cal A$ that was obtained using the standard theoretical technique of canonical paths.

\section{Conclusions and Future Work}
This paper makes two primary contributions. The first is the investigation of Markov Chain methods for uniformly sampling graphs with a fixed joint degree distribution. Previous work shows that the mixing time of $\cal A$ is polynomial, while our experiments suggest that the mixing time of $\cal B$ is also polynomial. 
The relationship between the mean of an edge and the autocorrelation values can be used to efficiently experiment with larger graphs by sampling edges with mean between 0.4 and 0.6 and repeating the analysis for just those edges. This was used to repeat the experiments for larger graphs and to provide further convincing evidence of polynomial mixing time.

Our second contribution is in the design of the experiments to evaluate the mixing time of the Markov Chain. In practice, it seems the stopping time for sampling is often chosen without justification. Autocorrelation is a simple metric to use, and can be strong evidence that a chain is close to the stationary distribution when used correctly.

\paragraph*{Acknowledgments} The authors would like to acknowledge helpful contributions of David Gleich, Satish Rao, Jaideep Ray, Alistair Sinclair, Virginia Vassilevska Williams and Wes Weimer.

\begin{small}
\bibliographystyle{plain}
\bibliography{jdd}
\end{small}

\section{Appendix}
\paragraph{Designing Synthetic Data}

Our goal was to represent all of the potential means for $\frac i {20}$ for $0<i\leq20$. We note that $20$ factors into 4 and 5, so we want to first fix some degrees such that $\dv_k=4$ and $\dv_l=5$. For convenience, because the maximum number of edges we will be assigning is 20, we will pick these degrees to be $K=\{20,21,22,23,24\}$ for $\dv_k=4$ and $L=\{25,26,27,28\}$ for $\dv_l=5$. The number of each we picked was to guarantee that there were at least 20 combinations of edge types. We can now assign the values $1-20$ arbitrarily to $\jdm_{K\times L}$. This assignment clearly satisfies that $\jdm_{k,l} \leq \dv_k\dv_l$ so far.

Now, we must fill in the rest of $\jdm$ so that $\dv$ is integer valued for degrees. One way is to note that we should have $4\times 20$ degree 20 edges. We can sum the number of currently allocated edges with one endpoint of degree 20, call this $x$ and set $\jdm_{1,20}=80-x$. There are many other ways of consistently completing $\jdm$, such as assigning as many edges as possible to the $K\times K$ and $L\times L$ entries, like $\jdm_{20,21}$. This results in a denser graph. For the synthetic graph used in this paper, we completed $\jdm$ by adding all edges as $(1,20),(1,21)$ etc edges. We chose this because it was simple to verify and it also made it easy to ignore the edges that were not of interest.
\end{document}